\documentclass[final]{siamltex}
\usepackage{amsmath, amssymb, graphicx}
\newcommand{\Eq}[1]{(\ref{eq:#1})}

\newcommand{\Def}[1]{Def.~\ref{def:#1}}
\newcommand{\Sec}[1]{\S \ref{sec:#1}}
\newcommand{\Fig}[1]{Fig.~\ref{fig:#1}}

\newcommand{\App}[1]{App.~\ref{app:#1}}

\newcommand{\InsertFig}[4]
{\begin{figure}[h!tb]
       \centerline{
         \includegraphics[width=#4]{./#1}
       }
       \caption{{\footnotesize  #2}
       \label{fig:#3}}
\end{figure}}

\newcommand{\InsertFigTwo}[5] {
\begin{figure}[h!tb]
       \centerline{
         \includegraphics[width=#5]{./#1}
         \hskip 0.05\linewidth
         \includegraphics[width=#5]{./#2}
       }
       \caption{{\footnotesize  #3}
       \label{fig:#4}}
\end{figure}}

\newcommand{\InsertFigFour}[7] {
\begin{figure}[h!tb]
       \centerline{
\renewcommand{\arraystretch}{0.01}
        \begin{tabular}{cc}
        \includegraphics[width=#7]{./#1}  &  \includegraphics[width=#7]{./#2} \bigskip \\
        \includegraphics[width=#7]{./#3}  &  \includegraphics[width=#7]{./#4}
        \end{tabular}
       }
       \caption{{\footnotesize  #5}
       \label{fig:#6}}
\end{figure}}

\newcommand{\bR}{{\mathbb{ R}}}

\newcommand{\bS}{{\mathbb{ S}}}

\newcommand{\cH}{{\cal H}}
\newcommand{\cI}{{\cal I}}
\newcommand{\cL}{{\cal L}}

\newcommand{\cP}{{\cal P}}
\newcommand{\cQ}{{\cal Q}}

\newcommand{\cS}{{\cal S}}
\newcommand{\cU}{{\cal U}}
\newcommand{\cV}{{\cal V}}

\newcommand{\eps}{\varepsilon}
\newcommand{\vphi}{\varphi}

\newcommand{\Wu}{W^{\mathrm{u}}}
\newcommand{\Ws}{W^{\mathrm{s}}}

\newcommand{\beq}[1]{\begin{equation}\label{eq:#1}}
\newcommand{\eeq}{\end{equation}}

\newcommand{\bsplit}[1]{\begin{equation}\label{eq:#1}\begin{split}}
\newcommand{\esplit}{\end{split}\eeq}

\title{Transport in Transitory Dynamical Systems \thanks
      {
        BAM and JDM were supported in part by NSF grant DMS-0707659.
        Useful conversations with Robert Easton, Doug Lipinski, and Troy Shinbrot are gratefully acknowledged.
      }}
\author{B.~A.~Mosovsky $^\dag$ and J.~D.~Meiss 
        \thanks{Department of Applied Mathematics, 
        University of Colorado, Boulder, CO 80309-0526 
        ({\tt brock.mosovsky@colorado.edu, james.meiss@colorado.edu})}
}
\date{\today}
\begin{document}

\maketitle

\begin{abstract}
We introduce the concept of a ``transitory'' dynamical system---one whose time-dependence is confined to a compact interval---and show how to quantify transport between two-dimensional Lagrangian coherent structures for the Hamiltonian case. This requires knowing only the ``action" of relevant heteroclinic orbits at the intersection of invariant manifolds of ``forward" and ``backward" hyperbolic orbits. These manifolds can be easily computed by leveraging the autonomous nature of the vector fields on either side of the time-dependent transition. As illustrative examples we consider a two-dimensional fluid flow in a rotating double-gyre configuration and a simple one-and-a-half degree of freedom model of a resonant particle accelerator. We compare our results to those obtained using finite-time Lyapunov exponents and to adiabatic theory, discussing the benefits and limitations of each method.
\end{abstract}



\pagestyle{myheadings}
\thispagestyle{plain}
\markboth{B.~A.~MOSOVSKY AND J.~D.~MEISS}{TRANSPORT IN TRANSITORY DYNAMICAL SYSTEMS}

%%%%%%%%%%%%%%%
%%  Introduction
%%%%%%%%%%%%%%%
\section{Transitory Systems}\label{sec:Introduction}

Invariant manifolds have long been recognized as important structures that govern global behavior in dynamical systems.  Hyperbolic manifolds in particular, by their very definition, relate information about the exponential contraction and expansion of nearby trajectories within the flow and so play crucial roles in the dynamics of such systems, lending insight into the mechanisms by which chaos, mixing, transport, and other complex global phenomena occur.  Transverse intersections of stable and unstable manifolds give rise to lobes defining of packets of trajectories that exit or enter coherent structures or resonance zones bounded by pieces of the manifolds. Thus, tracking these lobes provides a means for quantifying flux between coherent structures in the flow.

The treatment of mixing and transport for aperiodically time-dependent flows, however, requires the development of new methods because the concept of invariance may be too strong  and may not even lead to physically relevant structures. One popular method, the finite-time Lyapunov exponent (FTLE) has been used extensively in recent years to compute local approximations of invariant manifolds and identify structures that remain coherent in the Lagrangian sense on some finite time interval \cite{Haller00, Shadden05, Lekien07}.  Another idea uses a nonautonomous analog of hyperbolic orbits, called a distinguished hyperbolic trajectory \cite{Miller97, Rogerson99, Ide02, Madrid09}, and a third technique identifies approximately invariant regions as eigenfunctions of the Perron-Frobenius operator \cite{Froyland09}.

While these Lagrangian frameworks for identifying coherent structures in nonstationary systems have seen broad application, using them to accurately quantify transport and mixing over finite timescales remains a difficult task.  A few recent studies have managed to give numerical estimates for finite-time transport and mixing within aperiodic time-dependent flows \cite{Rogerson99, Coulliette01, Cardwell08, Mendoza10}. In addition, the instantaneous flux across a ``gate surface" can also be estimated using only local Eulerian information \cite{Haller98, Balasuriya06}, and the results have been applied to geophysical flows defined both analytically and by discrete data sets.  One such application is the investigation of eddy-jet interactions, in which the strength of the interaction is quantified by the amount of fluid entrained by or ejected from the eddy \cite{Poje99, Haller97}.  However, the accuracy of these methods is tied to the rate of change of the Eulerian velocity field, and, moreover, the instantaneous flux through a surface does not provide an estimate for the finite-time transport between two disjoint coherent structures.

Since the mere identification of Lagrangian coherent structures (LCS) in fully aperiodically time-dependent systems is a challenging problem itself, we study here a simpler problem---the quantification of finite-time transport between coherent structures in two-dimensional systems undergoing a transition between two steady states. Our methods are Lagrangian, in the sense that they rely on knowing certain key trajectories, and as such they allow us to compute the transport between two bounded coherent structures within the flow, as opposed to knowing only the flux across a single gate surface.  Even though we restrict the time-dependence to a compact interval, this problem provides, we believe, some insight into the more general aperiodic case. 

The systems of interest to us here are stationary in the limits $t \to \pm \infty$; these were called {\em asymptotically autonomous} systems by Markus \cite{Markus53}. However, we will assume that the system is nonstationary only on a compact interval, and will refer to these systems as {\em transitory}:

\begin{definition}[Transitory Dynamical System] \label{def:transitory}
A \emph{transitory dynamical system of transition time $\tau$} is one that is autonomous except on some compact interval of length $\tau$, say $[\tau_p, \tau_f]$ with $\tau_f - \tau_p = \tau$.  Thus on a phase space $M$, a transitory ODE has the form
  \beq{transitoryODE}
    \dot x = V(x,t)  , \quad  
    		V(x,t) = \left\{\begin{matrix} P(x) & t < \tau_p \\
			                       F(x) & t > \tau_f \end{matrix}\right. ,
  \eeq
where $P: M \to TM$ is the past vector field, $F: M \to TM$ is the future vector field, and $V(x,t)$ is otherwise arbitrary on  the transition interval $[\tau_p, \tau_f]$. 
\end{definition}

Though $P$ and $F$ are assumed autonomous, they could have arbitrarily complicated dynamics.  The {\em transition time} $\tau \equiv \tau_f-\tau_p$ is, relative to the time scales of $P$ and $F$, an especially important parameter for \Eq{transitoryODE} and, without loss of generality, we may set $\tau_p = 0$ so that $\tau_f = \tau$. 

It is not hard to envision many physical situations to which transitory dynamics would pertain. For example, any dynamical system that depends upon some parameters, $\dot x = V(x;p)$, can be made transitory if the parameters become time-dependent, $p \to p(t)$, and are allowed to switch from one state to another over a time interval of length $\tau$. We will consider a simple model of a particle accelerator in \Sec{pendulum} for which this is the case, but more realistic models could also be used \cite{Edwards04}. A similar physical system---also Hamiltonian---corresponds to a point particle in a billiard whose boundary evolves over time; for example, an ellipse with eccentricity that changes over a compact interval from one value to another (periodically oscillating boundaries have been studied in a number of cases, e.g. \cite{Lenz09}). Ecological models could also be transitory if the environment undergoes a shift (e.g., leading to a change in carrying capacity); one could study the ``transport" between basins of different equilibria under such a shift. Another class of examples corresponds to the change in flow regimes for a fluid in which there is an instability that grows and saturates. This could be due to external forcing (e.g., Rayleigh-Bernard convection or Taylor-Couette flow), to a flow that is driven by chemical reactions (e.g., Marangoni flow), or to an instability leading to eddy creation (e.g., for quasigeostrophic flows \cite{Rogerson99}). Finally, many models of fluid mixing could be studied in transitory regimes; for example, laminar flow through a pipe with a finite number of bends between two straight sections could be considered transitory \cite{Khakhar87,Aref02} and a fluid-filled cavity \cite{Chien86, Leong89, Anderson06} could be driven with a transitory mixing protocol.

For more general aperiodically time-dependent systems, infinite-time information cannot be used to identify coherent structures: indeed, in many fluid mechanical and observational applications, the behavior of the system is known only on a finite interval. This requires alternate definitions of approximate invariant manifolds.  These can be defined by extending the vector field to an infinite-time domain; however, since the extension is not unique, neither are the resulting manifolds \cite{Haller98, Sandstede00, Duc08, Yagasaki08}. Nevertheless, if the time of definition is long compared to the local expansion rates, this nonuniqueness results only in exponentially small corrections. By contrast, for transitory systems, stationarity outside the transition interval $[0,\tau]$ implies the classical notions of stable and unstable invariant manifolds can be used---see below.  Thus, transitory systems provide an aperiodic, time-dependent setting for which unique invariant manifolds exist. The advantages of using this asymptotic information will become clear when contrasting our methods with those employing FTLE in \Sec{RDG} and \Sec{pendulum}. In any case, a more general vector field that is defined only on a finite interval, $[0,\tau]$, say, could be extended by adding autonomous past and future vector fields to give a transitory vector field. Thus the relative unimportance of the form of the vector field's extension, given a long enough interval of definition, also applies to our case.

While we are not aware of other studies of transport for transitory systems, some aspects of transport for asymptotically autonomous systems in which the forward and backward limits are identical,
\beq{asympAuto}
  \lim_{t \to \pm \infty} V(x,t) = G(x),
\eeq
have been considered by Wiggins and collaborators \cite{Malhotra98, Samelson07}. A well-studied case corresponds to the adiabatic limit, when $\tau$ is large compared to the dynamical time scales of $P$ and $F$ \cite{Kruskal62}. Indeed, when \Eq{transitoryODE} is Hamiltonian and adiabatic, then the actions of the frozen system $P$ (or $F$) should be approximately preserved.  However as is well-known, adiabatic invariance breaks down near separatrices of the frozen system \cite{Cary86, Neishtadt87} and cases in which these separatrices sweep through a large portion of the phase space during the transition are of most interest to us. As was first noted by Elskens and Escande \cite{Elskens91}, when the time-dependence is periodic the entire region swept by the separatrix becomes a ``lobe" in the adiabatic limit. Transport properties have been studied in this limit \cite{Kaper91, Beigie91a, Kaper92}; however, these researchers assume that the frozen system has a parameter dependent curve of hyperbolic equilibria and this is typically not true for \Eq{transitoryODE}.  Moreover, we will not assume that $\tau$ is large.

In this paper we are interested in quantifying the transport between coherent structures of the past and future vector fields $P$ and $F$ for two-dimensional systems of the form \Eq{transitoryODE}. To define these, it is natural to consider hyperbolic orbits of $P$ and $F$ and their stable and unstable manifolds since initial segments of these often define invariant or nearly invariant structures such as {\em resonance zones} \cite{MMP87, Easton91, Lomeli09b}.  However, determining which structures of $P$ and $F$ are relevant to the dynamics of the full nonautonomous vector field $V$ requires special attention.  

It is natural to think of the dynamics of \Eq{transitoryODE} as occurring on the extended phase space $M \times \bR$. We will assume that $V$ has a complete flow, $\vphi_{t_1,t_0}: M \to M$ for any $t_0, t_1 \in \bR$, where $\vphi_{t_1,t_0}$ maps a point from its position at $t=t_0$ to its position at $t=t_1$.  Then every point $(x,t) \in M \times \bR$ has an orbit 
\[
  \gamma(x,t_0) \equiv \{(\vphi_{t_1,t_0}(x),t_1): t_1 \in \bR\} \subset M \times \bR, 
\]
and such sets are invariant in the sense that for any $\tau \in \bR$, $\gamma(x,t) = \gamma(\vphi_{\tau,t}(x),\tau)$. Of course, this notion of invariance is not restrictive: any subset $S \subset M$ can be taken to be the time-$\tau$ slice through an invariant set $\gamma(S, \tau) = \{ (x,t): x \in \vphi_{t,\tau}(S), t \in \bR\}$.  More generally the time-$t$ slice of any set $N \in M \times \bR$ can be defined using the standard projection  $\pi: M \times \bR \to M$ by
\beq{slice}
	N_t \equiv \pi(N \cap \{(x,t):x \in M \}) \subset M.
\eeq
Thus $\gamma_t(S,t) = S$.

If $\Lambda$ is any invariant set of the past vector field $P$, then the set $\{(\Lambda,t): t \leq 0\} \subset M \times \bR$ is a \emph{backward-invariant} set in the extended phase space. For example, equilibria and periodic orbits of $P$ are slices of backward-invariant sets of $V$. Similarly any invariant set of $F$ is a time-$t$ slice of a \emph{forward-invariant} set of $V$ for each $t > \tau$. Since the nonautonomous portion of the dynamics of \Eq{transitoryODE} is assumed to occur on a compact interval, it can be effected by a map, the {\em transition map} $T: M \to M$, defined as
\beq{transitionMap}
	T(x) = \vphi_{\tau,0}(x).
\eeq
Consequently an invariant set $\Lambda$ of $P$ becomes $T(\Lambda)$ at time $\tau$ and thereafter evolves under $F$.  If the dynamics of $P$ and $F$ are known, then the only nontrivial work we must do is to characterize the map $T$.

In addition to equilibria and periodic orbits, the unstable manifolds of orbits of $P$ and  stable manifolds of orbits of $F$ are also slices of invariant sets for $V$. For example, let $\Wu(\Lambda, X)$ denote the unstable manifold of an invariant set $\Lambda$ under a vector field $X$; it is the set of points that are asymptotic to $\Lambda$ as $t \to -\infty$. Consequently, if $\Lambda$ is an invariant set of $P$, its unstable manifold is a slice of the unstable manifold of the invariant set $\gamma(\Lambda,0)$ of $V$:
\[
	 \Wu(\Lambda,P) = \Wu_t(\gamma(\Lambda,0)) \equiv \Wu_t(\gamma(\Lambda,0),V), \quad  t < 0.\footnote
%%%
{
   As indicated here, we often will omit the $V$ from the notation 
   as it represents the default evolution.
}
%%% 
\]
However, a stable manifold $\Ws(\Lambda,P)$ is \emph{not} a slice of a stable manifold for $V$ since, due to the transition, it does not in general consist of points asymptotic to the orbit of $\Lambda$ as $t \to \infty$.  Instead, since the forward dynamics for any $t > \tau$ is determined by $F(x)$, the stable manifolds of invariant sets of $F$ are time-$t$ slices of the stable manifolds for $V$.  The unstable manifolds of invariant sets of $F$ have little dynamical relevance to $V$.

An example is sketched in \Fig{transitoryLoops}. Here $p$ is a hyperbolic saddle for $P$, but its orbit under $V$, $\gamma(p,0)$, is not hyperbolic. Indeed, when the point $T(p)$ does not lie on a hyperbolic orbit of $F$, $\gamma(p,0)$ has no stable manifold. Nevertheless, it does have an unstable manifold $\Wu(\gamma(p,0))$ and the temporal slices of this manifold coincide with $\Wu(p,P)$ when $t<0$. Furthermore, this manifold is an invariant set of $V$ with $\Wu_\tau(\gamma(p,0)) = T(\Wu(p,P))$ and its subsequent structure is obtained by simply evolving this set with $F$. We will say that such an orbit is {\em backward hyperbolic}. Similarly, if $f$ is a hyperbolic fixed point for $F$ it is a {\em forward hyperbolic} orbit of $V$, but not generally hyperbolic under $V$. Slices of the stable manifold $\Ws(\gamma(f,\tau))$ agree with $\Ws(f,F)$ for each $t > \tau$.  We formalize these notions as follows:
\begin{definition} \label{def:bfHyperbolic}
  A time-$t_0$ slice of an invariant set $\Lambda$ of a transitory dynamical system is \emph{backward hyperbolic} provided $\vphi_{t,t_0}(\Lambda_{t_0})$ is hyperbolic under $P$ for all $t < 0$.  It is \emph{forward hyperbolic} provided $\vphi_{t,t_0}(\Lambda_{t_0})$ is hyperbolic under $F$ for all $t > \tau$.
\end{definition}

%%%%%
\InsertFig{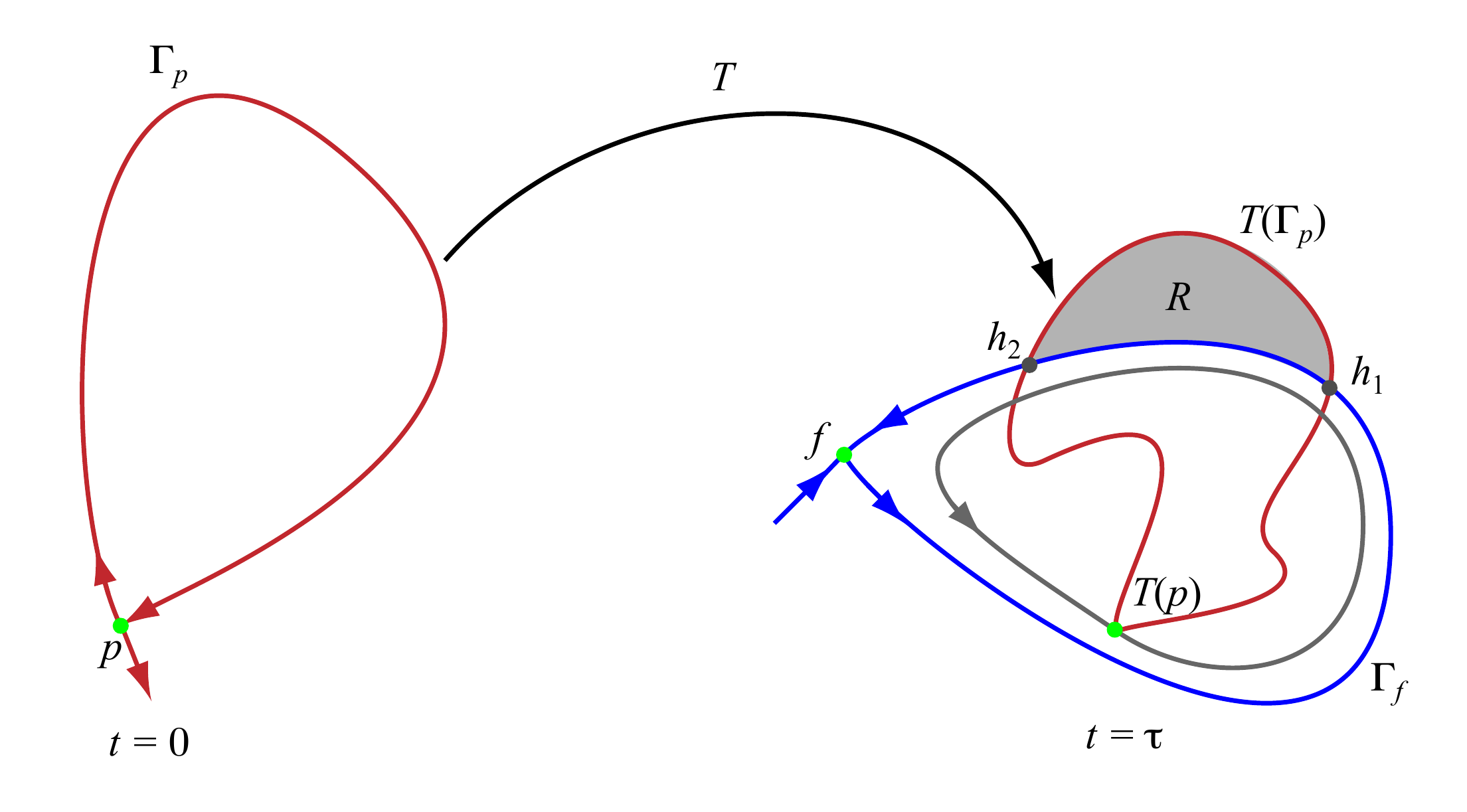}{A backward-hyperbolic orbit $\gamma(p,0)$ with homoclinic loop $\Gamma_p \subset \Wu(p,P)$ and a forward-hyperbolic orbit $\gamma(f,\tau)$ with homoclinic loop $\Gamma_f \subset \Ws(f,F)$. Here the unstable manifold of the orbit of $p$ intersects the stable manifold of the orbit of $f$ at $t=\tau$ at the heteroclinic points $h_1,\, h_2 \subset T(\Gamma_p) \cap \Gamma_f$ under a transitory flow.}{transitoryLoops}{4in}
%%%%%

In the simplest case, $\dim(M) = 2$ and $\Gamma_p \subset \Wu(p,P)$ and $\Gamma_f \subset \Ws(f,F)$ are homoclinic loops for $P$ and $F$, respectively, as sketched in \Fig{transitoryLoops}. The regions bounded by $\Gamma_p$ and $\Gamma_f$ can be thought of as {\em Lagrangian coherent structures} (LCS). If the former evolves under the map \Eq{transitionMap} to intersect the latter, then we can characterize the transport from one coherent structure to the other by these intersections.

Lagrangian coherent structures for nonautonomous systems are usually defined in terms of finite time stability exponents.  Haller and Yuan define them as regions bounded by ``material lines with locally the longest or shortest stability or instability time" \cite{Haller00}, while other researchers refer to LCS as curves or surfaces instead of regions, and identify these as ``ridges'' in the finite-time Lyapunov exponent field \cite{Shadden05, Mathur07, Branicki10}. In either case, since these coherent structures are defined only by finite time information, they may also persist only for some finite time. For \Eq{transitoryODE}, we will think of LCS as being objects bounded by separatrices of $P$ for $t<0$ or of $F$ for $t>\tau$.  These structures may also be ephemeral under the vector field $V$; however, for a transitory system we think of only one event, encapsulated by the transition map $T$, as creating or destroying an LCS. 

In \Fig{transitoryLoops}, the coherent structure bounded by $\Gamma_p$ in the past vector field $P$ is destroyed by the transition map $T$, giving rise to a new structure bounded by $\Gamma_f$ in the future vector field $F$. There are two heteroclinic points $\{h_1$, $h_2\} =  T(\Gamma_p) \cap \Gamma_f$ in the time-$\tau$ slice; they are backward asymptotic to $p$, forward asymptotic to $f$, and hence fully hyperbolic under $V$. Consequently, the set of orbits that begin inside $\Gamma_p$ but evolve to escape from $\Gamma_f$ is defined by the lobe $R$ bounded by the segments of $\Gamma_f$ and $T(\Gamma_p)$ between $h_1$ and $h_2$.  More generally, there may be more than one lobe and each lobe boundary may consist of more than two manifold segments, or equivalently, contain more than two heteroclinic points.  We discuss the computation of lobe areas in this general case in \Sec{ActionForm}.

For higher dimensional flows, even though the invariant manifolds $\Ws(p,P)$ and $\Wu(f,F)$ are not dynamically relevant for the full vector field $V(x,t)$, they may useful for defining {\em resonance zones} of $P$ and $F$, where a resonance zone with a ``small" escaping flux is a ``nearly" invariant set \cite{MMP87, Easton91, Lomeli09b}.  Particles initially in a resonance zone of $P$ when $t < 0$ will be approximately trapped up to time $0$, and those that find themselves in a resonance zone of $F$ at $t=\tau$ will be approximately trapped in the future. This makes low-flux resonance zones good candidates for coherent structures in the past or future vector fields.

In the remainder of the paper, we consider several simple examples of transitory systems for which the full vector field is a convex combination of the past and future vector fields:
\beq{transitoryV}
	V(x,t) = (1-s({t}))P(x) + s({t})F(x) .
\eeq
Here $s:\bR \to [0,1]$ is a \emph{transition function} satisfying
\beq{transitionFunction}
	 s(t) = \left\{ \begin{matrix} 0 &t < 0 \\ 
	                               1 &t > \tau
	                 \end{matrix} \right. ,	                             
\eeq
for transition time $\tau$.  While \Eq{transitoryV} is transitory in the sense of \Def{transitory} for \emph{any} function $s$ satisfying \Eq{transitionFunction}, for simplicity we usually take $s$ to be monotone nondecreasing.  Examples of such functions of varying smoothness are given in \App{transitionFunctions}.

%%%%%%% Traveling Wave Example %%%%%%%
As a first example, consider the traveling wave model studied by Knobloch and Weiss \cite{Knobloch87, Weiss89}.  This model for an incompressible, two-dimensional fluid is given in terms of the stream function
\bsplit{twStreamfunction}
	\psi(x,y,t) &= \psi_0(x,y) +\eps b(t) \psi_1(x,y),\\
	&\psi_0(x,y) = -cy + A\sin(kx) \sin(y), \\
	&\psi_1(x,y) = y,
\end{split}\eeq
on the domain $M = [0,2\pi]\times[0,\pi]$. The vector field is given by $V = \hat z \times \nabla \psi$, where $\hat z$ is the unit normal to the $xy$-plane, so the equations of motion are
\beq{StreamEqs}
	\dot x = -\frac{\partial}{\partial y} \psi ,\quad
	\dot y = \frac{\partial}{\partial x} \psi,
\eeq
This system is Hamiltonian with $(x,y)$ representing the coordinate and momentum and $H(x,y,t) = -\psi(x,y,t)$.

While Knobloch and Weiss studied the dynamics of $\psi_0$ with time-periodic perturbations, an asymptotically autonomous case of \Eq{twStreamfunction} was studied in \cite{Malhotra98, Samelson07}. The latter assumed that $\lim_{t\to \pm \infty} b(t) \to b_\infty \in \bR$ so that the past and future vector fields are equal.  For this case $b(t)$ is a ``bump function'' that plays the role of the transition function $s$ in \Eq{transitoryV}. If instead,  we assume that $b(t)$ has support only on the compact interval $[0,\tau]$, then \Eq{StreamEqs} is transitory in the sense of \Def{transitory}. One such function is 
\beq{bump}
         b(t) = s(t) \left(1-s(t)\right),
\eeq
where $s(t)$ is any transition function \Eq{transitionFunction}.

For \Eq{twStreamfunction} with \Eq{bump}, $P$ is identical to $F$, and when $|c|<|A|$ there are two hyperbolic equilibria on each of the lines $y=0$ and $y=\pi$ (see \Fig{travelingWave}).  The four  saddles of the past and future vector fields are denoted $p_i$ and $f_i$, respectively. Note that as subsets of $M$, $p_i = f_i$; however, as subsets of the extended phase space, they are points on different temporal slices and so their evolution under $V$ is distinct. For the simple perturbation $\psi_1$, the orbits of these points under $V$ remain on their respective horizontal lines.  Moreover, if $b(t) \ge 0$, then $\vphi_{t,0}(p_{1,2}) \to f_1$ and $\vphi_{t,0}(p_{3,4}) \to f_3$  as $t \to \infty$.

%%%%%
\InsertFig{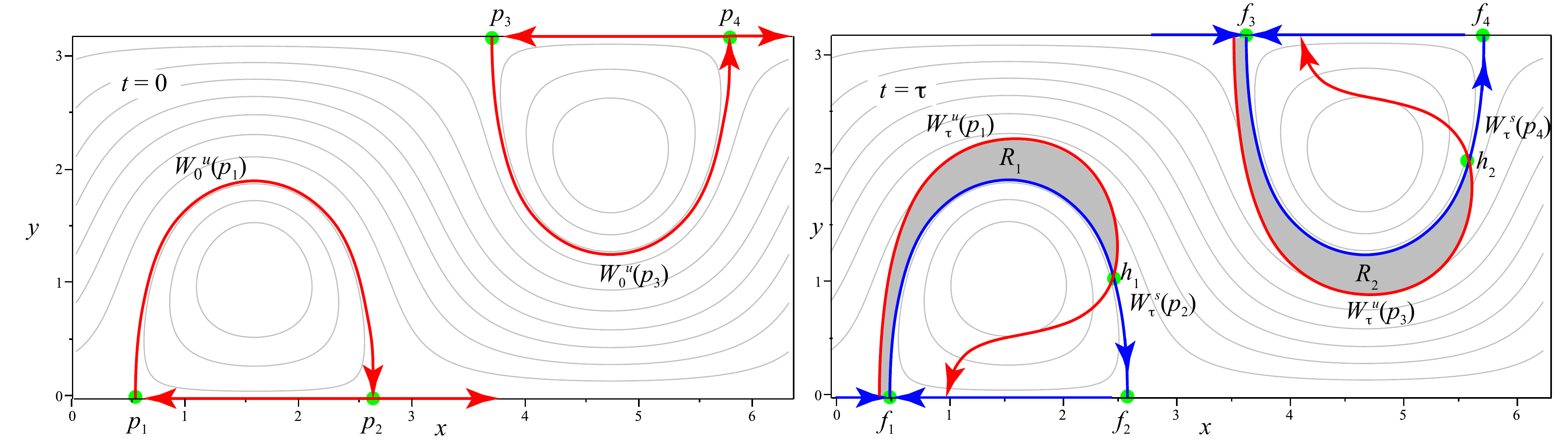}{Frozen time stream function at $t=0$ (left) and $t=\tau$ (right), and coherent structures for \Eq{twStreamfunction} with $A=k=\eps=1$ and $c=0.5$. The slices of the unstable manifolds in the right pane are obtained by numerically computing the transition map for $\tau = 4$ using the bump function \Eq{bump} with $s(t) = s_1(t/\tau)$, the cubic transition function in \Eq{polynomial}.}{travelingWave}{\linewidth}
%%%%%

The vector fields $P$ and $F$ each have two resonance zones in $M$ and these form the coherent structures of interest for the system \Eq{twStreamfunction}. For $P$, the two resonance zones should be thought of as being bounded by unstable manifolds since these will be slices of unstable manifolds of $V$; one is bounded by branches of $W^u(p_1, P) \cup W^u(p_2,P)$ and the other by $W^u(p_3,P) \cup W^u(p_4,P)$, as shown in \Fig{travelingWave}. For $F$ the resonance zones are bounded by stable manifolds: $W^s(f_1,F) \cup W^s(f_2,F)$ and $W^s(f_3,F) \cup W^s(f_4,F)$.  The unstable manifolds at $t=0$ evolve according to \Eq{StreamEqs}, and for the case shown in \Fig{travelingWave} numerical integration indicates that these intersect the stable manifolds at two heteroclinic orbits,
\beq{twHeteroclinic} \begin{split}
  h_1(t) &= W^u_t(p_1) \cap W^s_t(f_2), \\
  h_2(t) &= W^u_t(p_3) \cap W^s_t(f_4).
\end{split} \eeq
The lobes formed by these intersections, labeled 
$R_1$ and $R_2$, correspond to the trajectories that begin \emph{inside} the resonance zones of $P$ and end \emph{outside} the resonance zones of $F$.  By calculating the areas of these lobes we can quantify transport into and out of the coherent structures in the phase space, giving insight into the global dynamics of the system.

%%%%%%%%%%%%%%%%%
%% Flux
%%%%%%%%%%%%%%%%%
\section{Transitory Flux: Hamiltonian Case}\label{sec:Flux}
Coherent structures, by their very definition, denote regions of the phase space in which nearby trajectories behave similarly; examples include vortices or recirculation regions in fluid flows and resonance zones in the phase space of Hamiltonian systems. Since these coherent structures are typically composed of a large number of trajectories and since their boundaries separate dynamically distinct regions in the phase space, they can provide quite a bit of information about the global behavior of the system. In particular, knowing the incoming and exiting flux helps paint a global picture of the Lagrangian effects of time dependence within the vector field.  For example, for a particle accelerator we may wish to quantify the phase space volume corresponding to stable acceleration of particles (see \Sec{pendulum}).  In this case, computing the flux between coherent structures of the pre- and post-acceleration vector fields gives the desired quantity.

We proceed in this section to derive formulas for the flux between regions of phase space bounded by invariant manifolds of nonautonomous, one degree-of-freedom Hamiltonian systems (i.e. a system with $1\tfrac12$ degrees of freedom). The formulas obtained are the generalizations to the nonautonomous case of the action-flux formulas of \cite{MMP84}. For autonomous flows, these formulas were first obtained in \cite{MM86, MM88, MacKay91} and the nonautonomous case was studied in \cite{Kaper91} in the adiabatic limit.  

We note that \cite{Samelson07, Malhotra98} do compute lobe areas for asymptotically autonomous vector fields, in the sense of \Eq{asympAuto}.  However their theory relies on  $P = F$ and that these vector fields have a saddle equilibrium with a homoclinic trajectory. In the theory we present here, the saddles of $P$ and $F$ need not be the same and there need not be a homoclinic orbit of either autonomous vector field.  Instead we focus our attention on heteroclinic orbits of the time-dependent vector field $V$.

%%%%%%%%%%%%%%%%%
%% Action Formulas
%%%%%%%%%%%%%%%%%
\subsection{Flux by the Lagrangian Action}\label{sec:ActionForm}

Here we will compute the flux for a $1\frac12$ degree-of-freedom Hamiltonian vector field $V =  (\partial_yH, -\partial_xH)$ for a  nonautonomous Hamiltonian $H(x,y,t)$. More formally, if $\omega$ is a symplectic form,\footnote
%%%
{ Our notation is given in \App{notation}.}
%%%
e.g. $\omega = dx \wedge dy$, then the Hamiltonian vector field is determined by $\imath_V\omega = dH$. In this section we do not need to assume that the system is transitory, but we do assume that the phase space $M$ is two-dimensional and that the symplectic form is exact.

As discussed in \Sec{Introduction}, computing the flux corresponds to finding the area of some closed and bounded region $R \subset M$. Using Stokes' theorem, the resulting two-dimensional integral can be immediately reduced to an integral over the boundary:
\beq{stokes}
	\mbox{Area}(R) = \int_R \omega  = -\int_{\partial R} \nu ,
\eeq
where  $\nu$ is the Liouville one-form defined by $\omega = -d\nu$; for example, $\nu = ydx$. When $R$ is bounded by segments of stable and unstable manifolds, the flux formulas of \cite{MMP84} reduce integrals of the form \Eq{stokes} to action differences between orbits lying at the endpoints of the manifold segments. The action is given by an integral of the phase space Lagrangian, $L: M \times \bR \to \bR$
\beq{lagrangian}
	L(x,y,t) = \imath_V\nu -H(x,y,t) ,
\eeq
or, $L = y\dot{x} -H$ with $\dot{x}(x,y,t) = \partial_y H(x,y,t)$.

The simplest case is sketched in the left pane of \Fig{lobes}. Suppose that $\gamma_f$ is a forward hyperbolic orbit, $\gamma_p$ is a backward hyperbolic orbit and $R$ is a region in the time-$\tau$ slice that is bounded by a stable-unstable pair of segments of time-$\tau$ slices, $\cS \subset \Ws_\tau(\gamma_f)$ and $\cU \subset \Wu_\tau(\gamma_p)$, that intersect only on their boundaries, $h_0$ and $h_1$. Choosing the orientation of $\cU$ and $\cS$ consistent with a counterclockwise traversal of the boundary of $R$ gives $\partial R = \cU + \cS$ and $\partial \cS = -\partial \cU = h_0-h_1$.  Then \Eq{stokes} yields
\beq{area}
	\mbox{Area}(R) = -\left( \int_{\cS} \nu + \int_{\cU} \nu \right) .
\eeq
Therefore, to find the area we may compute the integral of the Liouville form along segments of stable and unstable manifolds.  The following two lemmas give formulas for computing these integrals in terms of the phase space Lagrangian \Eq{lagrangian} and the heteroclinic orbits $h_i(t) \equiv \vphi_{t,\tau}(h_i)$.
 
%%%%%
\InsertFig{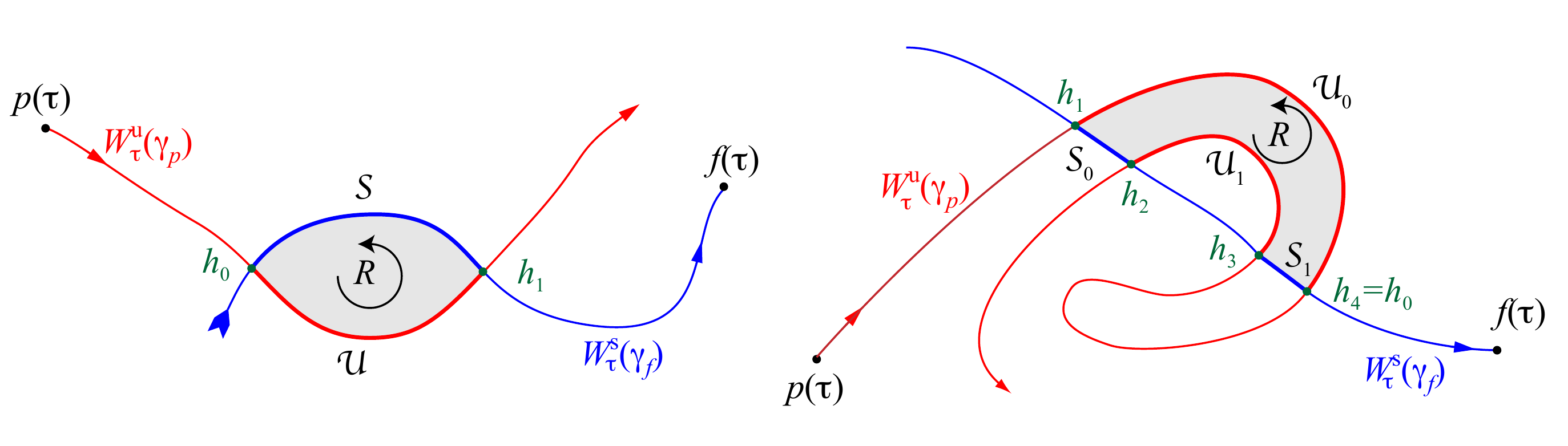}{Sketches of time-$\tau$ slices of lobes formed by stable and unstable manifold segments $\cS$ and $\cU$.  In the left pane the lobe is bounded by a pair of segments, and in the right pane by two pairs. Each segment is bounded by a pair of bi-asymptotic heteroclinic points $h_i$ and $h_{i+1}$.}{lobes}{5.5in}
%%%%%

%%%
\begin{lemma}\label{lem:PastAction}
Suppose that the orbits of $h_0$ and $h_1$ are backward hyperbolic and backward asymptotic, and that $\cU \subset \Wu_\tau(\gamma(h_i,\tau))$ is the time-$\tau$ slice of the unstable manifold that connects these points, with $\partial \cU = h_1- h_0$. Then
\beq{DeltaAMinus}
	 \int_{\cU} \nu = \Delta A^{-}_\tau(h_0,h_1) \equiv \int_{-\infty}^\tau 
	 	\big{[}L(h_1(s),s) - L(h_0(s),s)\big{]} ds .
\eeq
\end{lemma}
%%%

\begin{proof}
If we differentiate $\nu$ along $V$, Cartan's homotopy formula \Eq{Cartan} and 
\Eq{lagrangian} gives
\[
	\frac{d}{dt} \nu = \cL_V \nu = -\imath_V \omega + d(\imath_V\nu) = d L ,
\]
where $\cL_V$ is the Lie derivative \Eq{LieDerivative}.
Integrating this from $t$ to $\tau$, using \Eq{naturally} gives, for any $t$,
\[
	 \nu - \vphi^*_{t,\tau} \nu = \int_{t}^\tau \frac{d}{ds} \vphi_{s,\tau}^*\nu  ds = \int_{t}^\tau  d(\vphi_{s,\tau}^* L) ds .
\]
Consequently
\beq{flux} \begin{split}
	 \int_{\cU} \nu &= \int_{t}^\tau \left(\int_{\cU} d(\vphi_{s,\tau}^* L)\right) ds 
			+\int_{\cU}\vphi^*_{t,\tau} \nu \\
		&= \int_{t}^\tau \left(\int_{\vphi_{s,\tau}(\cU)} dL\right) ds +\int_{\vphi_{t,\tau}(\cU)} \nu \\
		&= \int_{t}^\tau \big{[} L(h_1(s),s) - L(h_0(s),s) \big{]} ds + \int_{\vphi_{t,\tau}(\cU)} \nu .
\end{split}
\eeq
Since $h_0(t) \to h_1(t)$, the length $|\vphi_{t,\tau}(\cU)| \to 0$ as $t \to -\infty$. Taking this limit yields the result \Eq{DeltaAMinus}.
\end{proof}

We note that $\Delta A_\tau^-(h_0,h_1)$ in \Eq{DeltaAMinus} is the difference between the ``past actions" of the orbits of $h_0$ and $h_1$.  A similar result, with an important sign change, holds for stable segments $\cS \subset \Ws_\tau(\gamma(f,\tau))$ connecting $h_1$ to $h_0$. In this case we replace $\cU$ by $\cS$ and let $t \to +\infty$ in \Eq{flux} to obtain

\begin{lemma}\label{lem:FutureAction}
Suppose that the orbits of $h_0$ and $h_1$ are forward hyperbolic and forward asymptotic, and that $\cS \subset \Ws_\tau(\gamma(h_i,\tau))$ is the time-$\tau$ slice of the stable manifold that connects these points, with $\partial \cS = h_0 - h_1$. Then
\beq{DeltaAPlus}
	\int_{\cS} \nu 
		= -\Delta A^{+}_\tau(h_1,h_0) \equiv -\int_\tau^{\infty} \big{[} L(h_0(s),s) 
			- L(h_1(s),s) \big{]} ds.
\eeq
\end{lemma}

Here $\Delta A_\tau^+(h_1,h_0)$ is the difference between the ``future actions'' of the orbits of $h_0$ and $h_1$.  The algebraic area of the lobe $R$ in the left pane \Fig{lobes} can be obtained by plugging the results \Eq{DeltaAMinus} and \Eq{DeltaAPlus} into \Eq{area}, yielding the action difference,
\beq{lobeArea}
	\mbox{Area}(R) = \Delta A (h_1,h_0) \equiv \int_{-\infty}^{\infty} \big{[} L(h_0(s),s) 
			- L(h_1(s),s) \big{]} ds.
\eeq
Note that $\tau$ appears nowhere in this result; indeed, the lobe area is independent of the time at which it is measured since the Hamiltonian flow is area-preserving.

It is important to note that equation \Eq{DeltaAMinus} can be used to compute the area under \emph{any} segment of an unstable manifold between two \emph{backward-asymptotic} points $h_0$ and $h_1$.  That is, the points need not be bi-asymptotic as in \Fig{lobes}.  For example, if $h_0$ is replaced with $p(\tau)$, equation \Eq{DeltaAMinus} can be directly used to integrate the Liouville form $\nu$ along the initial segment of $\Wu_\tau(\gamma_p)$ from $p(\tau)$ to $h_1$. Similar generality holds for equation \Eq{DeltaAPlus}. Moreover, the derivations of \Eq{DeltaAMinus} and \Eq{DeltaAPlus} do not require that the system be transitory; they are valid for any nonautonomous Hamiltonian vector field.

Lobes can have a more complicated structure, even in the autonomous case \cite{RomKedar90c}; in general, a lobe at time $\tau$ is a region $R$ bounded by an alternating sequence of stable segments of $\gamma_f$ and unstable segments of $\gamma_p$ whose intersections are topologically transverse.\footnote
%%%%%%%%%% Footnote: Non-topologically-transverse intersections
{If a heteroclinic point $h$ arises at an intersection that is not topologically transverse, the two boundary segments adjacent to $h$ can be combined (they necessarily have the same stability type) and $h$ can thus be ignored when calculating the lobe area.}
%%%%%%%%%% End FootnoteEquations 
An example with two pairs of segments is shown in the right pane of \Fig{lobes}.
The general formula for the area of such a lobe follows easily from 
\Eq{DeltaAMinus} and \Eq{DeltaAPlus}.
Suppose that $R$ is bounded by $2N$ such segments, and label the heteroclinic points $h_i$, $i = 0,\ldots,2N-1$ in a counterclockwise ordering on $\partial R$ setting $h_{2N} \equiv h_0$.  Without loss of generality, suppose the segment joining $h_0$ and $h_1$ is a portion of unstable manifold, call it $\cU_{0}$. Again using a counterclockwise ordering, label the next segment $\cS_{0}$, followed by $\cU_{1}$, etc., so that $\partial R$ consists of the alternating sum of $N$ unstable  and stable segments $\cU_{i}$ and $\cS_{i},\; i = 0,\ldots, N-1$.  The orientation of these segments is chosen to be consistent with the counterclockwise boundary so that  $\partial \cU_{i} = h_{2i+1} - h_{2i}$ and $\partial \cS_{i} = h_{2i+2} - h_{2i+1}$. Note that these orderings, as sketched in \Fig{lobes}, are not necessarily the same as ordering along the manifolds $\Wu$ or $\Ws$. Using \Eq{DeltaAMinus} and \Eq{DeltaAPlus} for each of the integrals along $\partial R$, gives the lobe area
\beq{2NLobeArea}
  \mbox{Area}(R) = \sum_{i=1}^{N} \Delta A(h_{2i-1}, h_{2i}) = -\sum_{i=0}^{N-1} \Delta A(h_{2i}, h_{2i+1}) ,
\eeq
where $\Delta A$ is defined by \Eq{lobeArea}.
While the two sums in \Eq{2NLobeArea} are trivially identical, the first can be thought of as the sum of the action differences between orbits bounding segments of \emph{stable} manifold along $\partial R$ and the second as the negative sum of action differences between orbits bounding segments of \emph{unstable} manifold.

Lemmas~\ref{lem:PastAction} and \ref{lem:FutureAction} apply to arbitrary nonautonomous systems provided only that the points bounding the segments $\cU$ and $\cS$ are past or future asymptotic, respectively. For the special case of transitory systems, the integrals in \Eq{DeltaAMinus} and \Eq{DeltaAPlus} can be further simplified since the phase space Lagrangian is autonomous outside $(0,\tau)$:
\[
	L(x,y,t) = \left\{ \begin{array}{cc} L_P(x,y)  & t \le 0 \\
					     L_F(x,y)  & t \ge \tau
			   \end{array} \right. .
\]
For example consider the lobe depicted in the left pane of \Fig{lobes} so that $h_0(0)$ and $h_1(0)$ both lie on $\Wu(p,P)$, the unstable manifold of a hyperbolic fixed point $p = p(0)$ of $P$. Since this manifold is stationary for all $t \leq 0$, these points are merely time-shifts of one another under the flow of $P$.  If we suppose that $h_1(0)$ is further along $\Wu(p,P)$ from $p$ than $h_0(0)$ (as in \Fig{lobes}), then there exists a $t_P < 0$ such that $h_1(t_P) = h_0(0)$ and thus
\[
  h_1(t_P + \alpha) = h_0(\alpha) \quad \forall \alpha \in (-\infty, 0\, ] .
\]
Consequently, the integral in \Eq{DeltaAMinus} reduces to two integrals over compact intervals:
\beq{transitoryDeltaAMinus}
  \int_{\cU} \nu = \Delta A^{-}_\tau(h_0,h_1) = \int_{t_P}^0 L_P(h_1(s)) ds + \int_0^\tau \big{[}L(h_1(s),s) - L(h_0(s),s)\big{]} ds .
\eeq
Similarly, assuming that points $h_0$ and $h_1$ are oriented along the stable manifold as in \Fig{lobes}, there is a $t_F > \tau$ such that $h_0(t_F) = h_1(\tau)$.  Then \Eq{DeltaAPlus} reduces to
\beq{transitoryDeltaAPlus}
  \int_{\cS} \nu = -\Delta A^{+}_\tau(h_1,h_0) = -\int_\tau^{t_F} L_F(h_0(s)) ds .
\eeq
Combining these yields the simplified form of the area \Eq{lobeArea} for the case of transitory systems:
\beq{transitoryArea}
	\mbox{Area}(R) =\int_0^{t_F} L(h_0(s),s) ds - \int_{t_P}^\tau L_P(h_1(s)) ds .
\eeq
The formulas \Eq{transitoryDeltaAMinus} and \Eq{transitoryDeltaAPlus} prove useful in numerical computations, and they can be applied to each action difference of \Eq{2NLobeArea} in the case of a more complicated lobe bounded by $2N$ manifold segments. We will use them in the examples of \Sec{Examples}.

%%%%%%%%%%%%%%%%%
%% Adiabatic
%%%%%%%%%%%%%%%%%
\subsection{Flux in the Adiabatic Limit}\label{sec:Adiabatic}

As in \Def{transitory}, a Hamiltonian $H(x,y,t)$ is transitory if
\[
	H(x,y,t) = \left\{ \begin{array}{ll}  H_P(x,y), &  t < 0\\
										 H_F(x,y), &  t > \tau
					 \end{array} \right. .
\]
If the transition time $\tau$ goes to infinity, adiabatic theory \cite{Kruskal62} may apply to such a system.  To consider this limit, it is helpful to reformulate the equations by introducing a ``slow time" variable
\[
	\lambda = \eps t ,
\]
where $\eps \equiv \tau^{-1}$ is to be thought of as small. Instead of thinking of $H$ as a function of time, through the transition function $s$, it is convenient to explicitly write it as a function of $s$,
\[
	\tilde H(x,y,s(\lambda)) = H(x,y,t) ,
\]
where $s$ is a transition function with transition time $1$. Now the past and future vector fields are generated by $H_P = \tilde H(x,y,0)$ and $H_F = \tilde H(x,y,1)$, respectively, and the vector field for $(x,y,\lambda) \in M \times \bR$ becomes
\bsplit{Adiabatic}
	\dot{x} &= \frac{\partial}{\partial y} \tilde H (x,y,s(\lambda)), \\
	\dot{y} &= -\frac{\partial}{\partial x} \tilde H(x,y,s(\lambda)),\\
	\dot{\lambda} &= \eps .
\end{split}
\eeq
The incompressible fluid flow \Eq{StreamEqs} is a transitory Hamiltonian system in this sense with $\tilde H(x,y,s) = -\psi(x,y,t)$.

For any finite $\eps$, the transition map \Eq{transitionMap} is now to be thought of as the time $\frac{1}{\eps}$ map
\[
	T(x,y) = \vphi_{\frac{1}{\eps}, 0}(x,y) .
\]
The \emph{frozen time} case corresponds to \Eq{Adiabatic} with  $\eps = 0$; alternatively it can be viewed as a family of autonomous systems on the phase space $M$ with Hamiltonians $\tilde H(x,y,s)$, $s \in [0,1]$.  

Adiabatic theory shows that, in certain cases, orbits of \Eq{Adiabatic} with $\eps \ll 1$ evolve so as to remain on orbits of the frozen time system with fixed \emph{loop action}, where the loop action for a closed curve $C$ is
\beq{loopAction}
  J(C) = \oint_C y\,dx ,
\eeq
i.e. the area enclosed.
Suppose there is a region $R \subset M$ in which every orbit of the frozen system is periodic and that there is a smooth family $\gamma_s \subset R$, $s \in [0,1]$ of periodic orbits of the frozen time systems with fixed action
\[
	J(\gamma_s) = J_0.
\]
Thus $\gamma_0$ is a periodic orbit of $P$, and $\gamma_1$ of $F$ and $J(\gamma_0) = J(\gamma_1)$. According to adiabatic theory, if $(x,y) \in \gamma_0$ and each of the $\gamma_s$ has a bounded period, then 
\beq{adiabaticity}
  d(T(x,y), \gamma_1) \to 0 \quad \textrm{as} \quad \eps \to 0,
\eeq
where $d(z,\Omega) =\inf_{\zeta \in \Omega} \| z-\zeta \|$ is the standard distance from a point $z$ to a set $\Omega$. 
The point is that $T$ approximately maps periodic orbits of $P$ to periodic orbits of $F$ with the same action, $T(\gamma_0) \approx \gamma_1$, when $\eps \ll 1$.

Adiabatic invariance breaks down if the frequencies of the periodic orbits in the family $\gamma_s$ are not bounded away from zero. This occurs, for example, when $\gamma_s$ approaches or crosses a separatrix of the frozen time system for some $s$. The resulting jumps in the action due to separatrix crossing were computed by \cite{Cary86, Neishtadt87}. 

The flux in the adiabatic limit was studied by Kaper and Wiggins \cite{Kaper91} under the assumption that the frozen systems $\tilde H(x,y,s)$ have a smooth, compact family of saddle equilibria, $p(s)$, with homoclinic loops $\Gamma_{p(s)} \subset \Wu(p(s)) \cap \Ws(p(s))$.  These authors state that the normally hyperbolic invariant manifold in the extended phase space,
\[
        \Lambda_0 = \{ (p(s(\lambda)),\lambda): \lambda \in \bR\} ,
\]
continues to a nearby normally hyperbolic invariant manifold, $\Lambda_\eps$, when $\eps$ is sufficiently small.  If the Melnikov function for \Eq{Adiabatic} has a nondegenerate zero, then the stable and unstable manifolds of $\Lambda_\eps$ intersect transversely for small $\eps$ defining a lobe  $R(\eps)$. Kaper and Wiggins show that the lobe area in the adiabatic limit becomes
\beq{AdiabaticLobe}
	\lim_{\eps \to 0} \mbox{Area}(R(\eps)) = J_{max} - J_{min},
\eeq
where
\[
	J_{max} = \max_{s \in [0,1]} J(\Gamma_{p(s)}) , \quad
	J_{min} = \min_{s \in [0,1]} J(\Gamma_{p(s)})
\]
are the maximum and minimum of the areas contained in the frozen time homoclinic loops.
The implication is that the region that is ``swept" by the separatrix is filled by the lobe, as was argued in \cite{Elskens91}. 

We will use \Eq{AdiabaticLobe} when possible in our examples below to discuss the limit $\tau \to \infty$; however, the assumption that the frozen time systems have a normally hyperbolic manifold of equilibria can be easily violated for a transitory system.

%%%%%%%%%%%%%%%%%
%% Stream Function Example
%%%%%%%%%%%%%%%%%
\section{Examples}\label{sec:Examples}

We will now consider several examples of transitory systems and use the  results of \Sec{Flux} to quantify transport between coherent structures.  For all examples we use the cubic transition function $s(t) = s_1(t/\tau)$ from \Eq{polynomial}, scaled so that $\tau$ is the transition time. Points along invariant manifolds are advected using a Runge-Kutta (5,4) Dormand-Prince pair. We represent the initial manifolds by a collection of equally-spaced points and add points using an adaptive interpolation method similar to that of Hobson \cite{Hobson93} when neighboring trajectories separate beyond a prescribed threshold. Similarly, trajectories are removed if the spacing decreases below a smaller threshold. We find the heteroclinic orbits from the advected manifolds using bisection to locate intersections of the stable and unstable invariant manifolds.

\subsection{Rotating Double Gyre} \label{sec:RDG}
The motion of a passive scalar in a two-dimensional incompressible fluid with the oft-studied double-gyre configuration is depicted in the left pane of \Fig{streamFig}. This configuration has been observed in both geophysical flows \cite{Poje99,Coulliette01} and experimental investigations of laminar mixing in cavity flows \cite{Chien86,Leong89}.  Here we consider a transitory flow that corresponds to a rotation of the two gyres by $\tfrac{\pi}{2}$ about $(\tfrac12,\tfrac12)$.  It is defined by the stream function
\bsplit{DoubleGyre}
	\psi(x,y,t) &= (1-s(t))\psi_P+ s(t) \psi_F  ,\\
	\psi_P(x,y) &= \sin(2\pi x)\sin(\pi y)  ,\\
	\psi_F(x,y) &= \sin(\pi x)\sin(2\pi y) ,
\end{split}
\eeq
and is Hamiltonian, with $H = -\psi$ and equations of motion \Eq{StreamEqs}.  Such a rotating regime could arise in a geophysical setting if the prevailing direction of the jet separating the two gyres changed over a finite time interval.  In terms of a cavity flow, \Eq{DoubleGyre} would model a regime in which a flow driven by upward movement of the left and right walls transitions smoothly to a flow driven by rightward movement of the top and bottom walls (cf. \Fig{streamFig}).  Taking a slightly different perspective, the transition in \Eq{DoubleGyre} could also be effected by modulating the flow boundary over the transition interval.  Such a changing boundary is used as a mixing mechanism in closed pipe flows, as investigated in \cite{Khakhar87,Aref02}, and is the driving force behind in-pipe mixing devices such as the Kenics$\textsuperscript{\textregistered}$ static mixer.

%%%%%
\InsertFigTwo{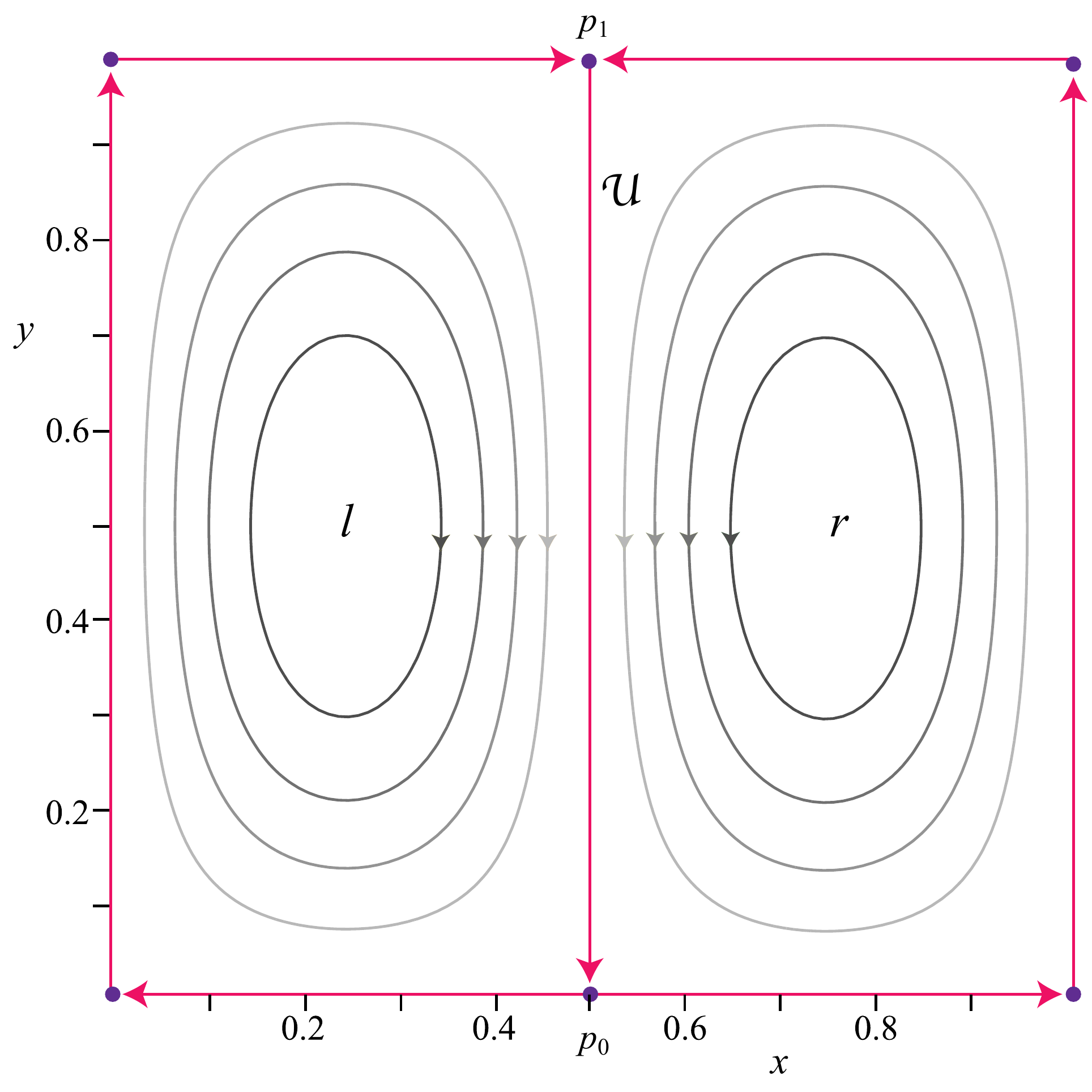}{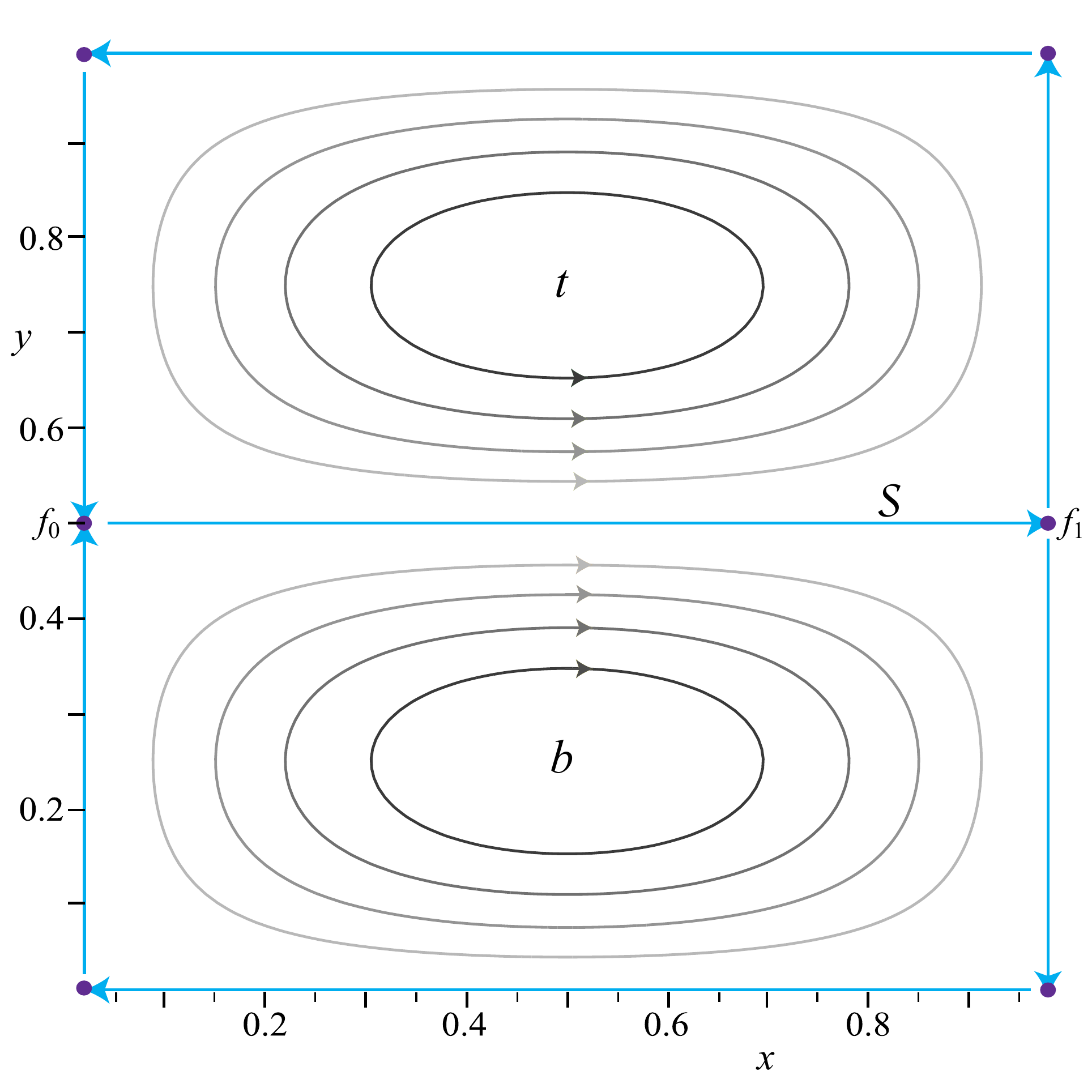}{Orbits of the stream functions $\psi_P$ (left) and $\psi_F$ (right) for \Eq{DoubleGyre}. The unstable manifolds for the saddles of $\psi_P$ are shown in red and the stable manifolds of the saddles of $\psi_F$ are shown in blue.}{streamFig}{0.475\linewidth}
%%%%%

The dynamics of \Eq{DoubleGyre} preserves the boundaries of the square $M = [0,1]\times[0,1]$ and each of the corners of $M$ is an equilibrium of $V$. The corners $(0,0)$ and $(1,1)$, are hyperbolic saddles for the full vector field $V$; however, though the corners $(1,0)$ and $(0,1)$ are saddles for both $P$ and $F$, they are \emph{not} hyperbolic under the full field $V$. For example, the unstable manifold of $(0,1)$ is a subset of its stable manifold; the slices of these manifolds at $t=0$ are
\[
	\Wu_0(0,1) = \{(x,1): 0<x<\tfrac12 \} \subset \Ws_0(0,1) = \{(x,1): 0<x<1\} .
\]
Similarly the stable manifold of $(1,0)$ is a subset of its unstable manifold;  the slices at $t=\tau$ are
\[
	\Ws_\tau(1,0) = \{(1,y): 0<y<\tfrac12\} \subset \Wu_\tau(1,0) = \{(1,y):0<y<1\} .
\]
Thus, though $(0,1)$ and $(1,0)$ are both forward and backward hyperbolic, in the sense of \Def{bfHyperbolic}, neither is a hyperbolic orbit of $V$.

The past vector field also has two saddle equilibria at $p_0 = (\tfrac12, 0)$ and $p_1 = (\tfrac12,1)$, and the future vector field has saddles at $f_0 = (0,\tfrac12)$ and $f_1 = (1,\tfrac12)$.  Under \Eq{DoubleGyre} the orbits of these points remain on the invariant boundaries of $M$ and, as we shall see, these orbits play crucial roles in the delineation of Lagrangian coherent structures for this system and the quantification of the flux between them.  

The natural coherent structures for the past vector field of \Eq{DoubleGyre} are the left and right gyres separated by the unstable manifold 
\beq{Udef}
	\cU = \{(\tfrac12,y): 0<y<1\}
\eeq
of $p_1$, see \Fig{streamFig}. Note that for any $t<0$, $\Wu_t(\gamma(p_1,0)) = \cU$. Moreover, for any $t$, the stable manifold of this orbit is $\Ws_t(\gamma(p_1,0)) = \{(x,1): 0\le x < 1\}$, and consequently $\gamma(p_1,0)$ is a hyperbolic orbit of $V$. For the future vector field the top and bottom gyres are coherent structures with separatrix given by the stable manifold
\beq{Sdef}
	\cS = \{(x, \tfrac12): 0<x<1\}
\eeq
of the future hyperbolic point $f_1$. Note that $\Ws_t(\gamma(f_1,\tau)) = \cS$ for $t > \tau$, and $\Wu_t(\gamma(f_1,\tau)) = \{(1,y): 0 \le y<1\}$.

%%%%%%
\InsertFigFour{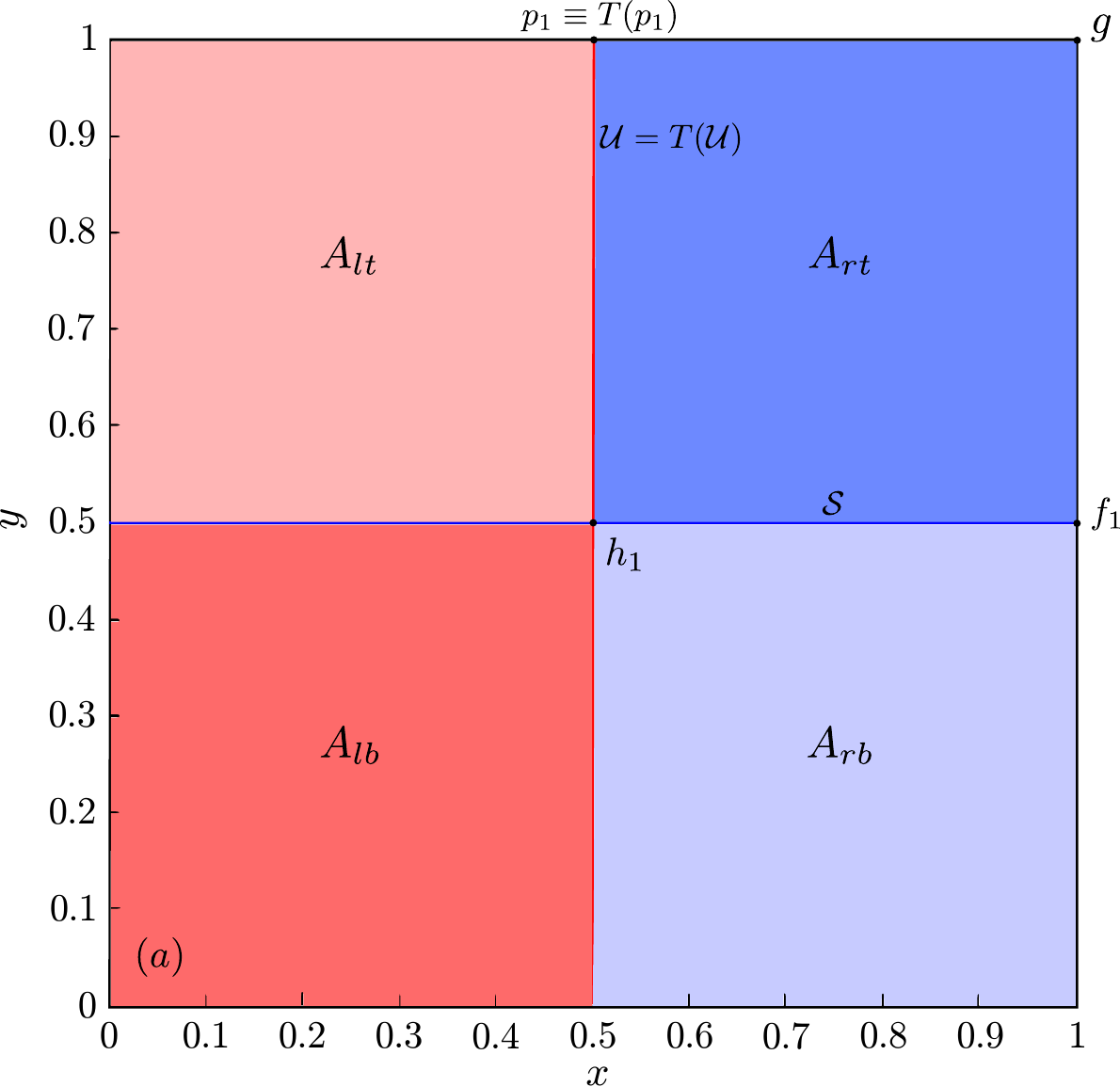}{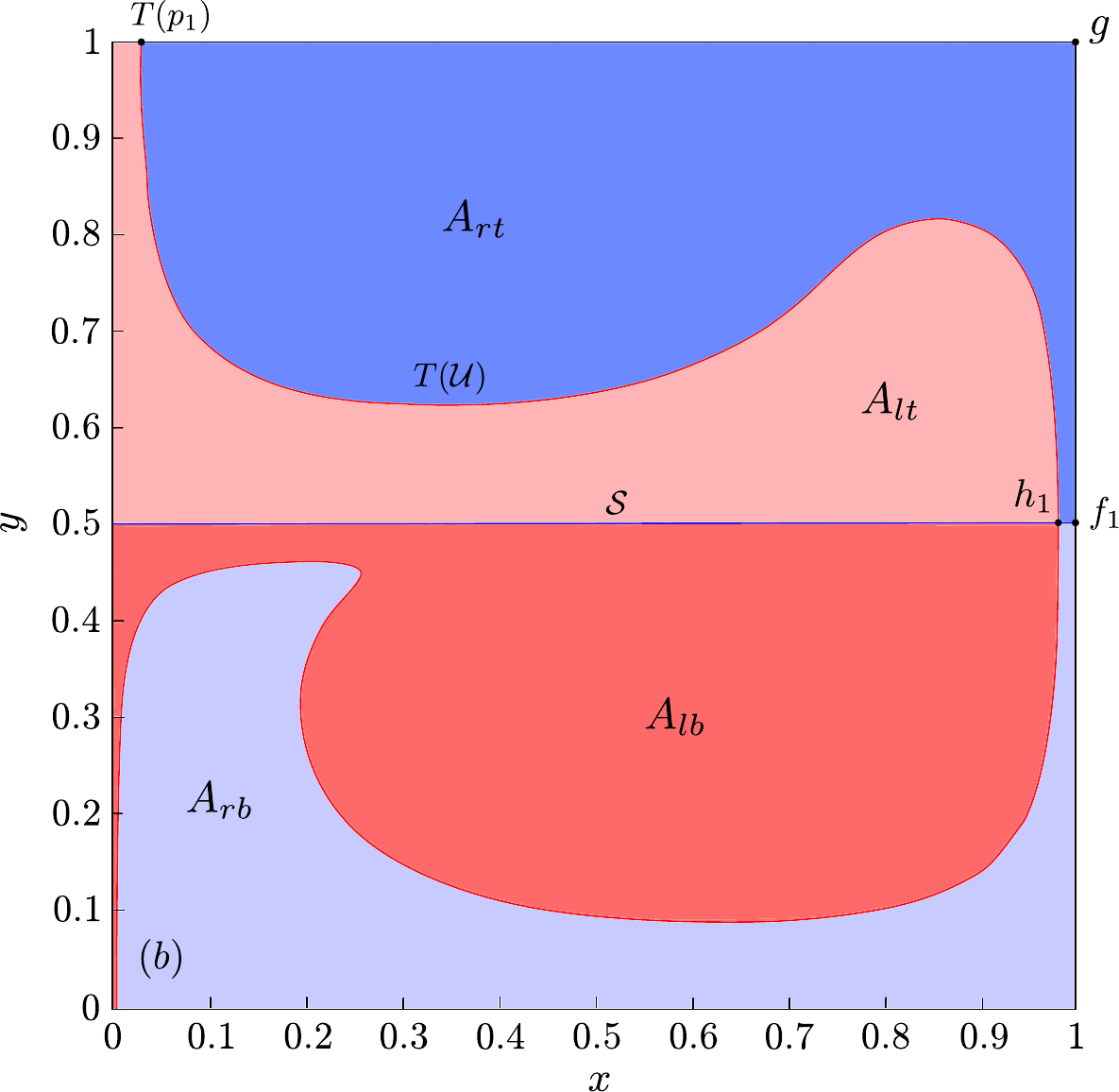}{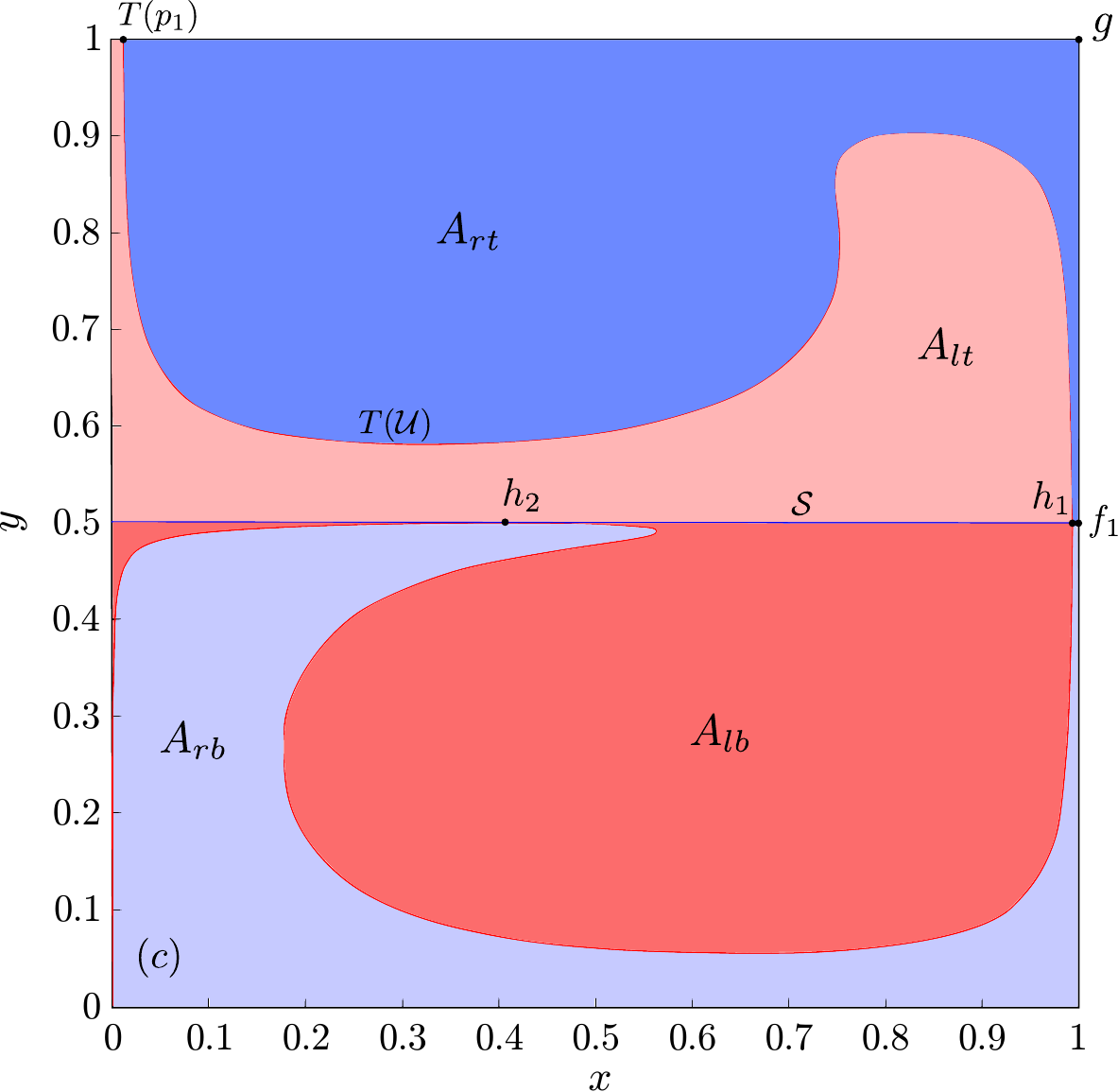}{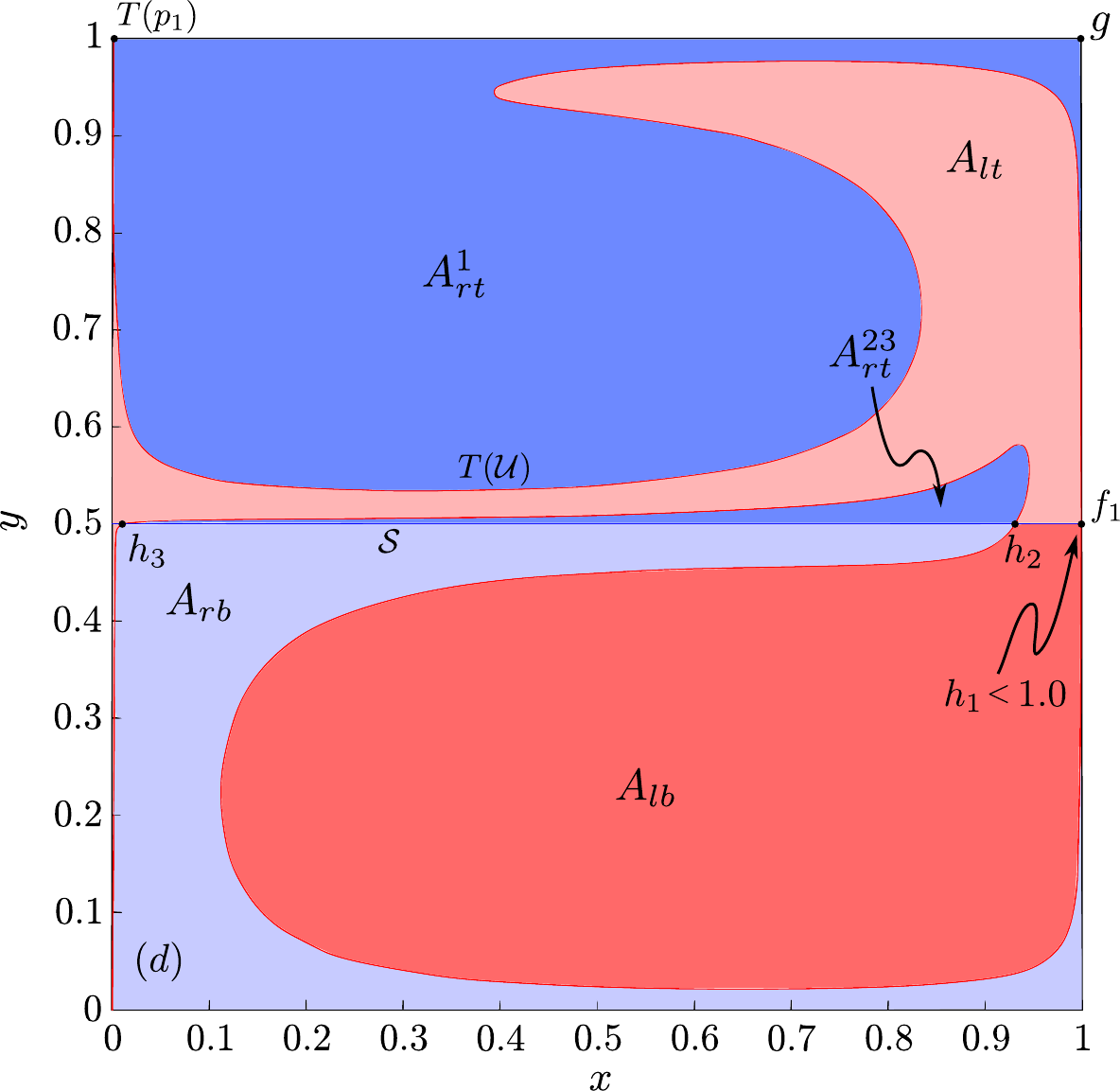}{$T(\cU)$ for transition times $\tau=0.0$ (a), $0.4$ (b), $0.5314$ (c) and $0.8$ (d) for the model \Eq{DoubleGyre}. Trajectories that begin in the left gyre are colored red, those that begin in the right gyre are colored blue, and the dividing curve between these is $T(\cU)$. The dark blue region, labeled $A_{rt}$, is the portion of the right gyre that ends in the top gyre at time $\tau$, and the dark red region, labeled $A_{lb}$, is the portion of the left gyre that ends in the bottom gyre. The heteroclinic points are labeled $\{ h_1,h_2,h_3 \} = T(\cU) \cap \cS$. A movie showing the variation of $T(\cU)$ as $\tau$ varies from  $0$ to $3.0$ is at [LINK MOVIE ``TDGManifolds.mov" HERE].  }{TDG}{0.475\linewidth}
%%%%%%

Transport between the two pairs of gyres is completely determined by the image of the past separatrix, $T(\cU)$, or equivalently, the preimage of the future separatrix, $T^{-1}(\cS)$.  A movie showing the evolution of the manifolds starting with $\cU$ and $T^{-1}(\cS)$ at $t=0$ and ending with $T(\cU)$ and $\cS$ at  $t = \tau = 0.7$ is at [LINK MOVIE ``TDGAdvection.mov" HERE]. However, to use the formulas of \Sec{Flux}, we only need to know the orbits heteroclinic from $\gamma(p_1,0)$ to $\gamma(f_1,\tau)$, and at $t=\tau$ these consist of points on $T(\cU) \cap \cS$. At least one such heteroclinic orbit is guaranteed since $\partial M$ is invariant under \Eq{DoubleGyre}; that is, the segment $T(\cU)$ still connects the top to the bottom and thus must cross $\cS$ an odd number of times. When $\tau = 0$ the transition map is the identity and the only intersection is $T(\cU) \cap \cS = h_1 = (\tfrac12, \tfrac12)$ (see the top left pane of \Fig{TDG}). This intersection persists as $\tau$ increases, simply moving to the right along $\cS$ and finally limiting on $f_1$ as $\tau \to \infty$. For the cubic transition function, we observe that $h_1$ is the only intersection providing $\tau \lesssim 0.5314$, at which point a new pair of heteroclinic points, $h_2,\, h_3$, are created in a saddle-node bifurcation, as shown in the bottom left pane of \Fig{TDG}. The next heteroclinic bifurcation occurs at $\tau \approx 3.6908$ and the creation of heteroclinic orbits accelerates as $\tau$ increases; indeed for $\tau = 5.0$ there are $19$ such heteroclinic orbits and the number appears to grow without bound as $\tau \to \infty$.

For the  model \Eq{DoubleGyre} it is easy to obtain the manifold structure at $t=0$ from that at $t=\tau$, as shown in \Fig{TDG}, because the system has a time-reversal symmetry whenever the transition function obeys the relation
\beq{transitionSymmetry}
	s(\tau-t) = 1-s(t) .
\eeq
Note that the cubic function we use and all of the polynomial transition functions given in \App{transitionFunctions} have this property. In this case, the dynamics is reversed by the involution $R: M \times \bR \to M \times \bR$ defined by
\beq{reversor}
	R(x,y,t)=(y,x,\tau-t) .                                                         
\eeq
It is easy to see that $\psi$ is invariant under this transformation and that the vector field is reversed by $R_*$, the push-forward \Eq{pushforwardV} of $R$:
\[
	DR\,  V(y,x,\tau-t) = -V(x,y,t) .
\]
Consequently $R$ inverts the transition map,
	$T^{-1} = R \circ T \circ R$,
and since $R(\cS) = \cU$,
\[
	T^{-1}(\cS) = R(T(\cU)) . 
\]
Consequently, the phase portraits of $\cU$ and $T^{-1}(\cS)$ at $t=0$ can be obtained from those in \Fig{TDG} at $t =\tau$ by reflection about $y=x$ and exchanging $\cU$ and $\cS$.

Denoting the left and right gyres of the past vector field by $l$ and $r$ and the top and bottom gyres of the future vector field by $t$ and $b$ respectively, there are four fluxes of interest, $A_{i j}$, corresponding to the trajectories starting in gyre $i \in \{l,r\}$ at time $0$ and ending in gyre $j \in \{t,b\}$ at time $\tau$ (see \Fig{TDG}). For example, $A_{rt}$ is the area of the region that is to the right of $\cU$ for $t \leq 0$ and above $\cS$ for $t \geq \tau$. Thus, at $t=\tau$ these regions are bounded by segments of $T(\cU)$ and $\cS$.  They can consist of multiple disjoint lobes when there are additional heteroclinic points, as in pane (d) of \Fig{TDG} in which there are two lobes for $A_{rt}$.

Since \Eq{DoubleGyre} is incompressible and $\partial M$ is invariant,
\[
	A_{l t} + A_{l b} = A_{r t} + A_{r b} = \tfrac12 ,
\]
the total area of a single gyre in the past and future vector fields.  Moreover, since the top and bottom gyres are filled completely by the images of the left and right gyres,
\[
	A_{l t} + A_{r t} = A_{l b} + A_{r b} = \tfrac12 .
\]
Consequently, knowledge of one of these four areas uniquely determines the remaining three.

To calculate $A_{r t}$ for a given transition time $\tau$ we must integrate the form $\omega$ over the the dark blue region in \Fig{TDG}, or equivalently, integrate the form $-\nu$ over its boundary.  We first consider the case of only one heteroclinic point $h_1 = T(\cU) \cap \cS$.  In this case, $A_{rt}$ is comprised of a single lobe and, given the points $T(p_1) = (x_p,1)$ and $h_1 = (x_h,\tfrac12)$, the integrals along the top, right, and bottom edges of this lobe are trivial.  Then, 
\beq{Art}
	A_{r t} = \frac12(1+x_h) - x_p -\int_{T(\cU_{p_1}^{h_1})} \nu ,
\eeq
where $T(\cU_{p_1}^{h_1})$ is the oriented segment of $T(\cU)$ from $\gamma_\tau(p_1,0) = T(p_1)$ to $\gamma_\tau(h_1,\tau) = h_1$.  According to \Eq{DeltaAMinus}, the last term in \Eq{Art} is simply the backward action difference $\Delta A^{-}_\tau(T(p_1),h_1)$ between the orbits of $h_1$ and $T(p_1)$.  Thus, to evaluate \Eq{Art} we must compute $T(p_1)$ and $h_1$, and finally integrate the Lagrangian along the backward asymptotic orbits of $h_1$ and $T(p_1)$ from $-\infty$ to $\tau$ to compute this action difference.

Computation of $T(p_1)$ is straightforward and is accomplished by numerically integrating the vector field $V$ over the transition interval $[0,\tau]$ with initial condition $p_1$.  Computing $h_1$ is slightly more involved and requires a root-finding algorithm to determine the intersection of $T(\cU)$ with $\cS$ at $t=\tau$.  We begin with points equally spaced along $\cU$ at $t=0$ and numerically integrate $V$ over $[0,\tau]$ to compute the image under the transition map $T$ of each point.  As the manifold stretches during advection, we adaptively refine it to maintain the initial resolution.  That is, at the first time step in which two neighboring trajectories diverge beyond a prescribed separation tolerance, we initialize a new trajectory at their midpoint and continue to track it over the remainder of the interval.  We also remove points in much the same manner in areas where the manifold is contracting, helping to speed up computation.  Upon obtaining $T(\cU)$, we bracket the intersection with $\cS$ and use a bisection routine to determine $h_1$ to the desired accuracy, initializing new trajectories where necessary.  It should be noted that since unstable manifolds attract orbits in forward time, this adaptive refinement is a stable process and we incur minimal numerical error in the resulting manifolds $T(\cU)$ \cite{Hobson93}. In fact, results for $T(\cU)$ obtained by refining the initial point spacing at $t=0$ and by adaptively refining during integration were virtually indistinguishable, with the adaptive computation being an order of magnitude faster in some cases.  Finally, the Lagrangian \Eq{lagrangian}, which for this system is
\bsplit{DGLagrangian}
	L(x,y,t) &= y \dot{x} + \psi(x,y,t) = (1-s(t))L_P(x,y) + s(t) L_F(x,y) \\
	      L_P(x,y)   &= \sin(2\pi x) \big[ \sin(\pi y)-\pi y \cos(\pi y) \big] ,\\
	      L_F(x,y)   &= \sin(\pi x) \big[ \sin(2\pi y)-2\pi y \cos(2\pi y) \big] ,
\end{split}
\eeq
is integrated along the computed orbits $\gamma(p_1,0)$ and $\gamma(h_1,\tau)$ using Simpson's rule.

When there are additional heteroclinic points (e.g., $h_2$ and $h_3$ in \Fig{TDG}, pane $(d)$), they can be computed in precisely the same way as $h_1$, described above. For the case of three heteroclinics the total flux is $A_{rt} = A_{rt}^1 + A_{rt}^{23}$, where the two terms on the right-hand side represent the areas of the two disjoint lobes.  The area of the larger lobe $A_{rt}^1$ is calculated using \Eq{Art} while the smaller lobe is similar to that shown in \Fig{lobes} and hence we calculate its area according to \Eq{lobeArea}:
\[
	A^{23}_{r t} = \Delta A(h_3,h_2) = \int_{-\infty}^{\infty} \big{[} 
		L(h_2(t),t) - L(h_3(t),t) \big{]} dt ,
\]
which is positive by the counterclockwise orientation of the segments $\cU_{h_2}^{h_3}$ and $\cS_{h_3}^{h_2}$.

Computation of the last term in \Eq{Art} is greatly simplified using \Eq{transitoryDeltaAMinus}, since the system is transitory, and by noting that for \Eq{DGLagrangian} $ L_P(\tfrac12, y) \equiv 0 $ for any $y$. Similarly, \Eq{transitoryArea} can be used to simplify the computation of $A^{23}_{rt}$.

Results for the computation of $A_{rt}$ for transition times ranging from 0 to 3.69 are summarized in \Fig{TDGflux}.  Note the increase in the rate of change of flux at $\tau \approx 0.531$ corresponding to the emergence of the second lobe of area $A_{rt}^{23}$.  At $\tau \approx 3.69$, a new pair of heteroclinic points $h_4$ and $h_5$ is created and their corresponding manifolds delineate a new lobe of trajectories initially to the left of $\cU$  for $t<0$.  Indeed, each new heteroclinic bifurcation as $\tau$ increases creates a new lobe that alternately adds area to $A_{rt}$ or to $A_{lt}$.

%%%%%%
\InsertFig{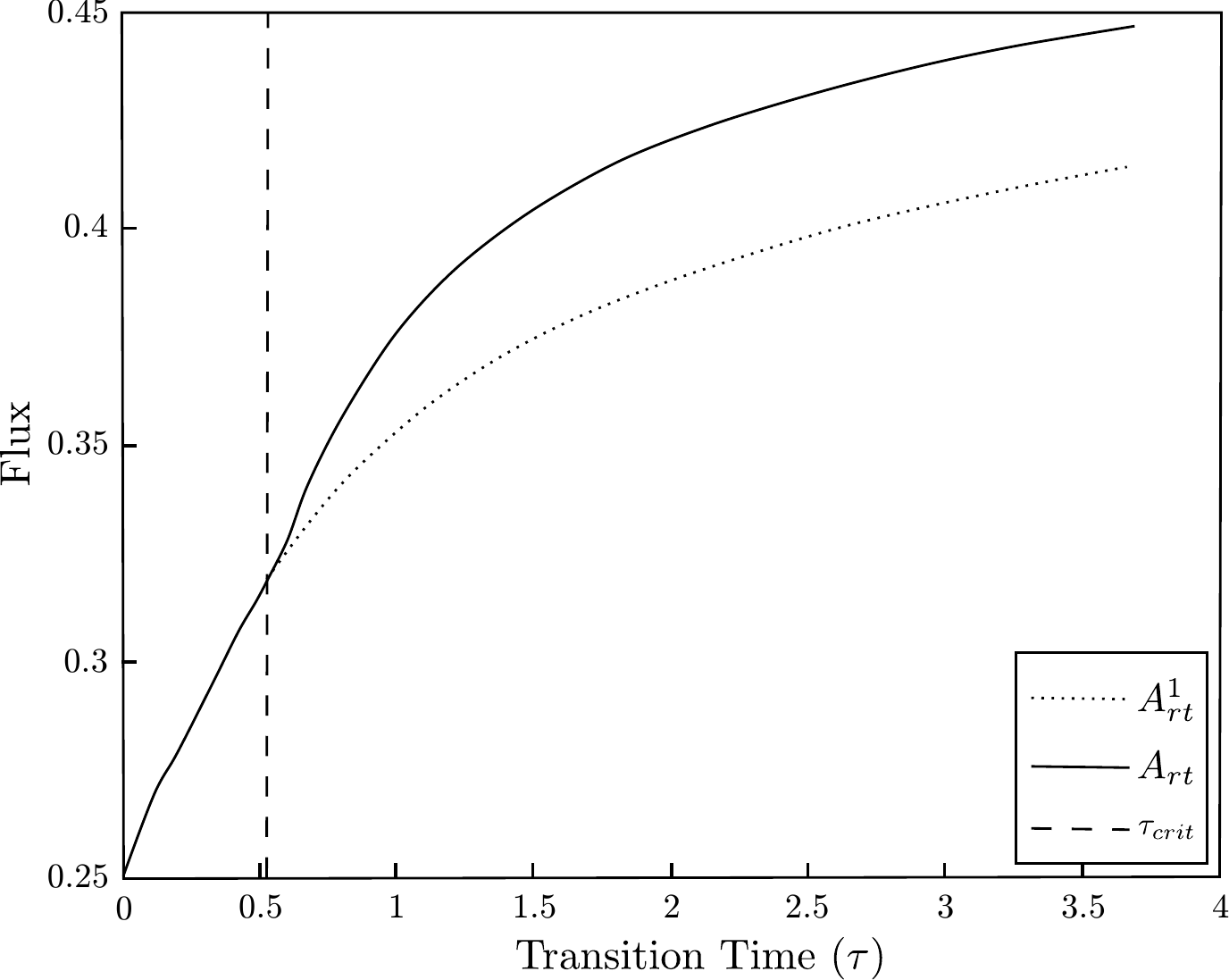}{Plot of the right-to-top flux, $A_{rt}$, as a function of transition time.  Note that a second lobe of area $A_{rt}^{23}$ emerges at $\tau_{crit} \approx 0.531$.  Its effect is manifested by the departure of the solid line from the dashed one for $A^1_{rt}$.}{TDGflux}{3in}
%%%%%%

%%%%%%%

As $\tau \to \infty$ some of the orbits of \Eq{DoubleGyre} can be described by the adiabatic theory outlined in \Sec{Adiabatic}. For example, each periodic orbit in the neighborhood of the elliptic equilibrium of the left gyre of $P$ continues to a periodic orbit of the frozen-time system---setting $\tilde \psi(x,y,s(t)) = \psi(x,y,t)$---with fixed loop action. As $s$ grows from $0$ to $1$ in \Eq{DoubleGyre} the left gyre rotates by $\frac{\pi}{2}$, so these orbits evolve continuously to periodic orbits of $F$ enclosing the elliptic equilibrium of the bottom gyre. If a family of periodic orbits of the frozen time system with fixed loop action remains a bounded distance away from the family of separatrices of $\tilde \psi(x,y,s)$, then the period of each orbit in this family is bounded, and as $\tau \to \infty$ the adiabatic theory implies that the actual evolution will follow that of the frozen system. The implication is that when $\tau \gg 1$, $T$ will approximately map periodic orbits of $P$ in the left gyre to periodic orbits of $F$ in the bottom gyre with the same action (and similarly for the right and top gyres). 

An indication of the approach to adiabaticity is displayed in \Fig{TDGAdiabatic}, which shows four elliptic orbits $\gamma^i_0,\: i = 1,...,4$ of $P$ (left pane) and their images under the transition map $T$ for two values of $\tau$ (middle and right panes). The dashed curves represent orbits $\gamma^i_1$ of $F$ having the same loop actions as the initial orbits.  When $\tau$ is small, as in the middle pane, each of the images $T(\gamma^i_0)$ differs visibly from $\gamma^i_1$. Conversely, when $\tau$ is moderately large, as in the rightmost pane, $T$ maps the innermost three $\gamma^i_0$ virtually on top of the corresponding $\gamma^i_1$. The outermost orbit, $\gamma^4_0$, (red curve) maps to a loop with tendrils far from $\gamma^4_1$ as its period is not sufficiently small for adiabaticity to pertain at this value of $\tau$. Violation of adiabaticity could also be seen in its most extravagant form by $T(\cU)$ itself, which, as noted above, necessarily crosses the square from top to bottom and intersects the line $\cS$ infinitely many times as $\tau \to \infty$. Nevertheless, since increasingly many orbits of the left gyre map to their counterparts in the lower gyre as $\tau$ increases, adiabatic theory implies that
\[
  A_{lb} = A_{rt} \to 0.5 \quad \textrm{as} \quad \tau \to \infty,
\]
as \Fig{TDGflux} seems to suggest.

While this result is clearly demonstrated by the numerical evidence shown in \Fig{TDGAdiabatic} and the corresponding movie linked in its caption, the theory of Kaper and Wiggins \cite{Kaper91} for the flux in the adiabatic limit, outlined in \Sec{Adiabatic} above, does not apply directly to the double-gyre model \Eq{DoubleGyre}.  There is no family of homoclinic loops for $\tilde \psi(x,y,s(\lambda))$, though, as Kaper and Wiggins themselves point out, the theory could be straightforwardly extended for a family of heteroclinic cycles, as would be appropriate for the double-gyre.  When $s < \frac 12$, the ``left" gyre of $\tilde \psi(x,y,s(\lambda))$ is bounded by separatrices connecting  four hyperbolic equilibria: $(0,0)$, $(0,1)$, and a saddle on each of the upper and lower boundaries.  This family of separatrices looses hyperbolicity at $s = \frac12$, when the point $(0,1)$ is no longer a hyperbolic equilibrium of the frozen system and the separatrix becomes a triangle. Though the point $(0,1)$ again becomes hyperbolic when $s > \frac12$, it does not appear that the result \Eq{AdiabaticLobe} can be rigorously applied.

%%%%%%
\InsertFig{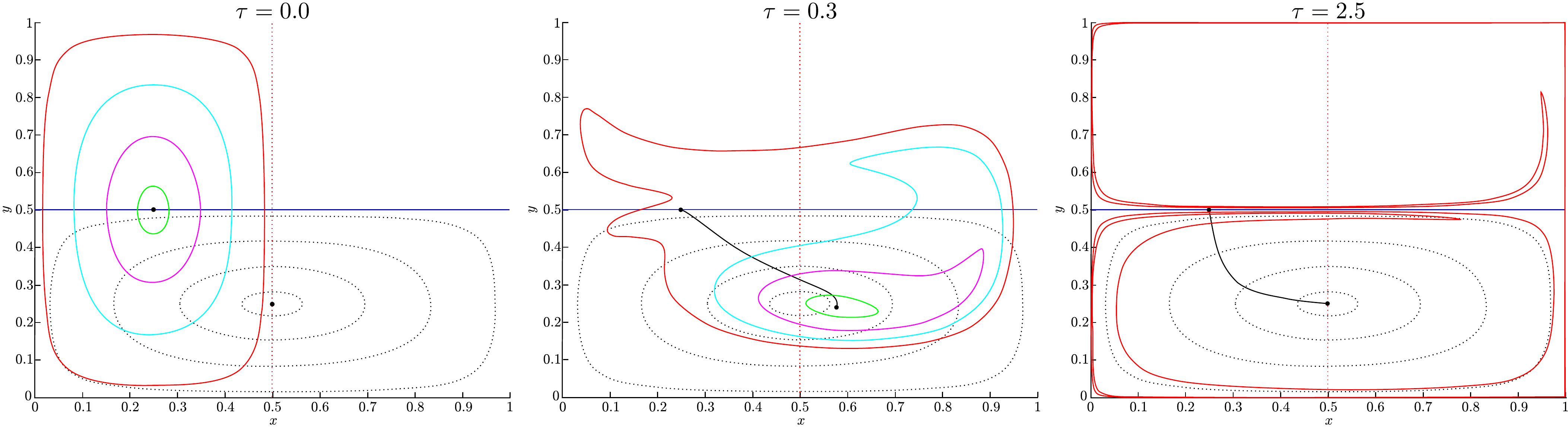}{Illustration of the approach to adiabatic invariance of periodic orbits for \Eq{DGLagrangian}. The solid curves depict the images of four orbits of $P$ under the transition map $T$ for $\tau= 0$ (left pane),  $\tau = 0.3$ (middle pane), and $\tau = 2.5$ (right pane). The dashed curves represent orbits of the future vector field with the same loop actions as the initial orbits. The solid black curve depicts the orbit of the left elliptic equilibrium of $P$ over $0 \le t \le \tau$. A movie showing the emergence of adiabatic behavior as $\tau$ increases is [LINK MOVIE ``TDGAdiabatic.mov" HERE]}{TDGAdiabatic}{\linewidth}
%%%%%%

Finally, we compare our mode of analysis with a technique commonly used for analyzing time-dependent flows: the finite-time Lyapunov exponent (FTLE).  We will comment only briefly here on the similarities and differences between these two methods in the context of identifying heteroclinic trajectories and computing lobe areas for \Eq{DoubleGyre}.  A more thorough explanation of the use of FTLE for approximating invariant manifolds is given in \cite{Shadden05, Haller01}.

%%%%%%
\InsertFig{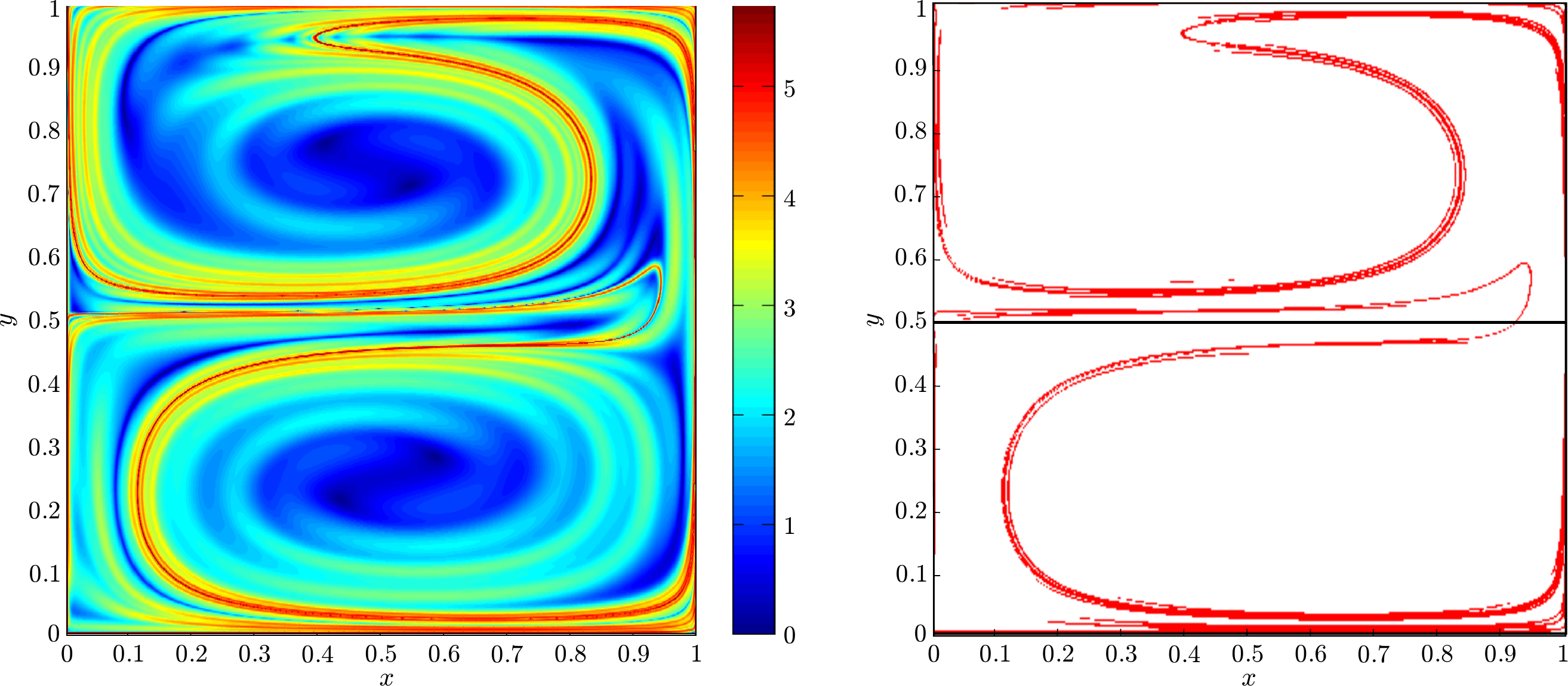}{ Backward time FTLE field for \Eq{DoubleGyre} at the transition time $\tau = 0.8$ using a backward integration time of 1.2 and a $1500 \times 1500$ grid. The left pane shows contours of the FTLE value from blue (smallest) to red (largest), and the right pane shows the ridge extracted by keeping only those values within 30\% of the maximum. This ridge gives an approximation to $T(\cU)$.}{dgFTLE}{\linewidth}
%%%%%%

The backward-time FTLE field for \Eq{DoubleGyre} with transition time $\tau = 0.8$, is shown in the left pane of \Fig{dgFTLE}. The ``ridges'' of the FTLE field approximate the unstable boundaries of LCS; the red regions in the figure.  Comparing \Fig{dgFTLE} with pane $(d)$ of \Fig{TDG}, one can see that the most prominent ridge of the FTLE field corresponds to the curve $T(\cU)$.  However, the global picture given by the FTLE field is complicated by secondary ridges near the main ridge.  These secondary ridges are common (see for example \cite{Shadden05, Haller01, Brunton10}) and, while it is unclear whether they are numerical anomalies or they offer physical insight into the stretching of nearby trajectories over the time-scale used for computation, they complicate the numerical extraction of the most prominent ridge and the identification of any heteroclinic orbits.

As noted in \cite{Shadden05}, the ridges of the FTLE field become ``more Lagrangian'' as the integration time grows.  Since the transitory system \Eq{DoubleGyre} has a trajectory $\vphi_{t,0}(p_1)$ that is truly backward hyperbolic (i.e. it does \emph{not} lose its backward hyperbolicity for any time $t \in \bR$), as the integration time increases we should expect the most prominent ridge of the backward-time FTLE to become increasingly aligned with the unstable manifold $T(\cU)$ of $T(p_1)$ at time $\tau$.  We do indeed observe this; however, the secondary ridges also become more pronounced, making extraction of the main ridge increasingly difficult.  This places a practical upper limit on the length of the approximate invariant manifold that can be computed using FTLE calculations, as the numerical extraction of the ``main'' ridge becomes infeasible for large integration times.  

Several numerical methods for efficiently extracting the appropriate ridges from the FTLE field have been proposed \cite{SP07, GGTH07, Lipinski10}; however, in practice, the most prominent ridges are typically extracted by simply filtering out all values below a prescribed threshold.  Such a filter with the threshold set at 70\% of the maximum FTLE value is shown in the right pane of \Fig{dgFTLE}. For the chosen integration time, the secondary ridges are so close in height to the main ridge that they can not be removed by this simple height filter.  That is, as the filtering threshold is increased, gaps appear in the main ridge before all the secondary ridges have disappeared.  Thus, while the main FTLE ridge qualitatively agrees with the true unstable manifold $T(\cU)$, an accurate computation of flux between coherent structures is difficult to obtain with this method.

%%%%%%%%%%%%%%%%%
%% Transitory Pendulum
%%%%%%%%%%%%%%%%%
\subsection{Resonant Accelerator}\label{sec:pendulum}

As a second example, we consider a system that serves as a highly simplified model of a particle accelerator \cite{Edwards04}. Here the coherent structures are ``resonances" that result in the trapping of particles in an accelerating potential well, and the goal is to determine the phase space region that represents stable acceleration. In our model, this corresponds to orbits that begin within a stationary, past resonance and ultimately end in a moving, future resonance at $t=\tau$. 

The basic model is given by a Hamiltonian of the form
\[
	H(q,p,t) = \frac12 p^2  + V(q-\theta(t)).
\]
We assume that the potential well is initially stationary, then accelerates, and eventually reaches a constant velocity so that the phase $\theta(t)$ obeys
\beq{theta}
	\theta(t) = \left\{ \begin{matrix}  0 & t< 0 \\  \omega t +\phi & t > \tau  
						\end{matrix} \right..
\eeq
While this system is not transitory in the sense of \Def{transitory} (note the time dependence of the potential function $V$), we can convert it to one that is with the canonical transformation 
\[
  (q,p,H) \mapsto (\cQ,\cP,\cH) = (q-\theta(t), p, H -p\dot \theta(t)).
\]
Note from \Eq{theta} that the time derivative of the phase is proportional to a transition function \Eq{transitionFunction}, namely $\dot \theta(t) = \omega s(t)$.
 Reverting to the original variable names gives the new Hamiltonian
\beq{accHamiltonian}
	H(q,p,t) = \frac12 p^2 - \omega s(t) p + V(q) ,
\eeq
and so the past and future systems are autonomous:
\beq{accAutonHam}
	H(q,p,t) = \left\{ \begin{array}{lll}
						H_P(q,p) =& \frac12 p^2 + V(q), & t < 0 \\
						H_F(q,p) =& \frac12 (p-\omega)^2 + V(q) -\frac12 \omega^2, & t > \tau
					   \end{array}\right. .
\eeq
Specifically, for our model we take
\beq{accPotential} \notag
  V(q) = -k \cos(2\pi q),
\eeq
so both autonomous limits are equivalent to the pendulum, with the resonance centered around $p=0$ for $t \leq 0$ and $p=\omega$ for $t \geq \tau$.  Without loss of generality we can scale variables to set $\omega=1$, leaving two parameters: the transition time $\tau$ and the potential energy amplitude $k$. All examples and parameter values below correspond to this scaled system.  
 
Contours of  $H_P$ and $H_F$ are shown in \Fig{accContours}, and since $\omega = 1$, the transition $H_P \to H_F$ corresponds to a unit vertical translation of the past vector field. For each $s$, the frozen-time Hamiltonian has saddles at $(q,p) = (\pm\frac12,s)$ with stable and unstable manifolds 
\beq{accSeparatrix}
  p_\pm(q,s) = \pm \sqrt{2k(1+\cos(2\pi q))} + s,
\eeq
that define separatrices bounding a resonance. Of course, by periodicity the two saddles can be identified, so that the manifolds actually correspond to homoclinic loops on the cylinder $M = \bS \times \bR$.  The past resonance corresponds to $s=0$ and its separatrices are denoted $\cU_\pm$ as these are slices of the unstable manifolds of the saddle for the full vector field when $t\le0$. The future resonance corresponds to $s=1$ and its separatrices are denoted $\cS_\pm$, as these are slices of stable manifolds of the saddle for the full vector field when $t \ge \tau$. 
The width of each resonance of the frozen time system is 
\beq{resWidth}
	w = p_+(0,s)-p_-(0,s) = 4\sqrt{k},
\eeq
which for this simple model is independent of $s$. Note that the past and future resonance zones ``overlap" when $w >1$, implying that $k > \frac{1}{16}$. 

%%%%%%%%%
\InsertFigTwo{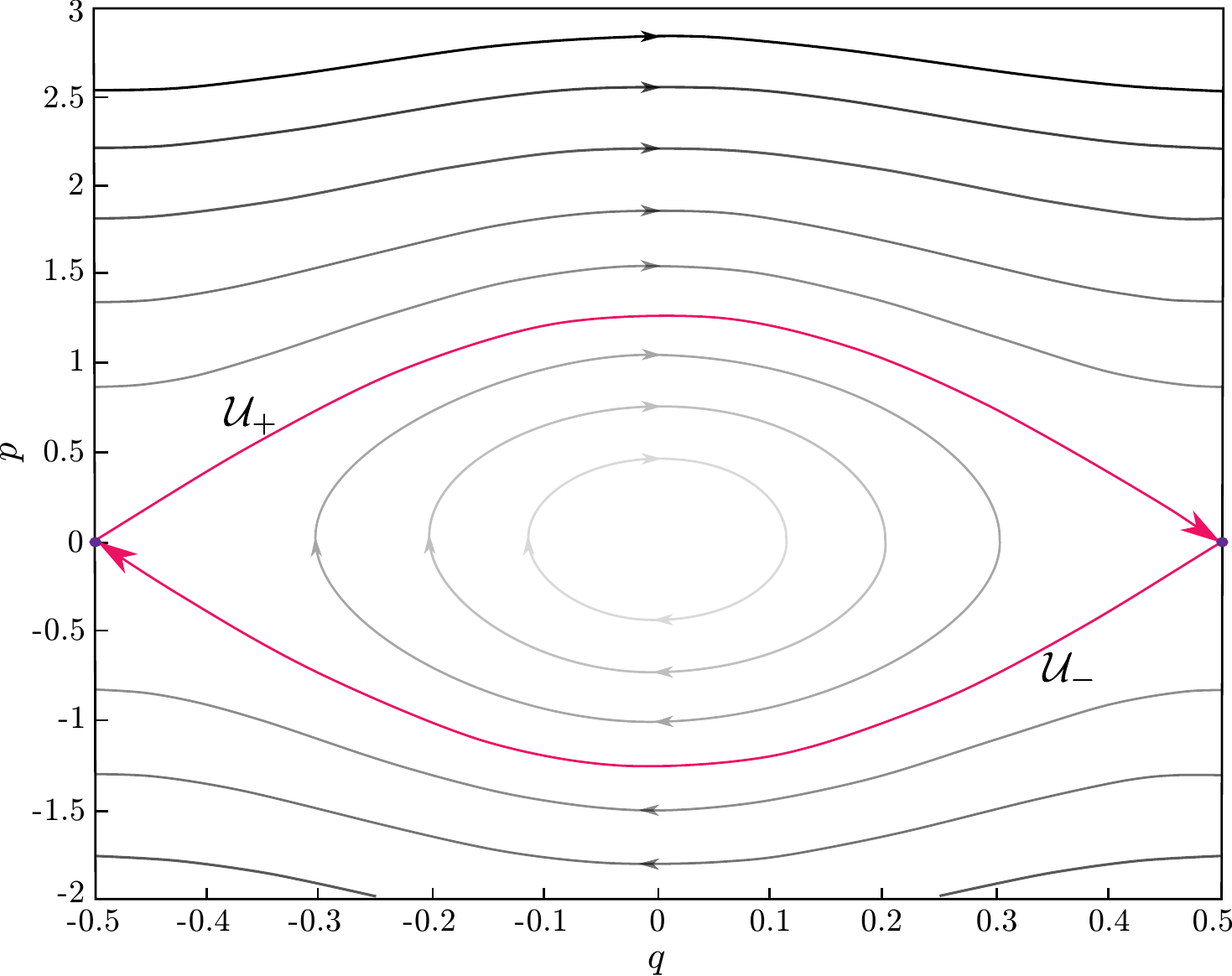}{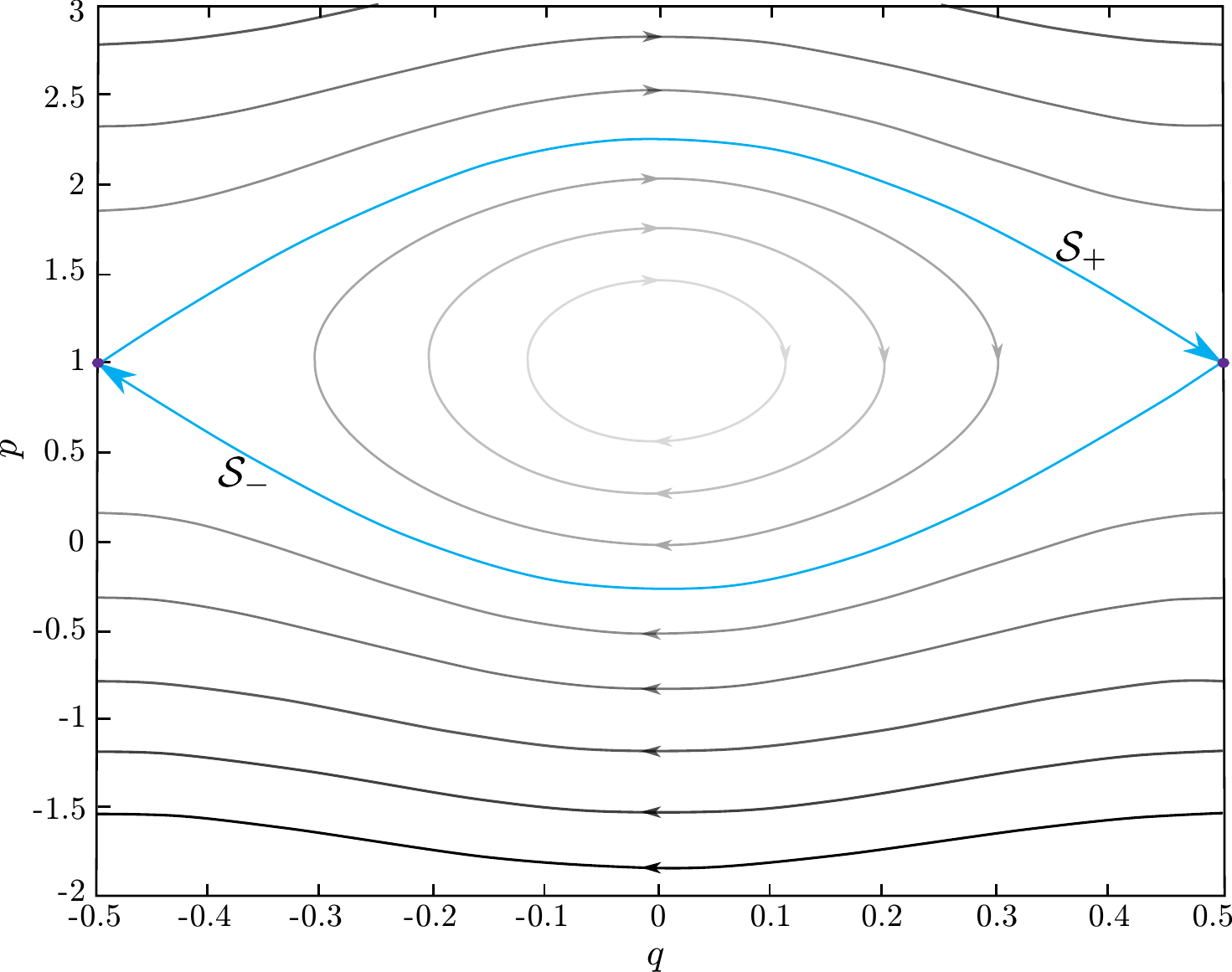}{Contours of the past (left) and future (right) Hamiltonians \Eq{accAutonHam} for $k=0.4$. The unstable manifolds for the saddles of the past vector field are shown in red and the stable manifolds for the saddles of the future vector field are shown in blue.}{accContours}{0.475\linewidth}
%%%%%%%%%

Let $A_{io}$ be the area of the region that begins \emph{inside} the past resonance at time $t=0$ and ends \emph{outside} the future resonance at time $t=\tau$, with corresponding notations $A_{ii}$, $A_{oi}$ and $A_{oo}$ for the other beginning and ending configurations.  We are principally interested in calculating the fraction of accelerated phase space area,
\beq{proportion} \notag
  R_{acc} = \frac{A_{ii}}{A_{ii} + A_{io}}.
\eeq
The images of the manifolds $\cU_\pm$ under the transition map for two values of $\tau$ are shown in \Fig{acc}. Here $A_{io}$ is the area of the region that is inside $T(\cU_\pm)$ and outside $\cS_\pm$. This region, light blue in the figure, appears disconnected; however, on the cylinder it is formed from a single connected set. The region $A_{ii}$ is dark blue in the figure and corresponds to particles that remain trapped within the resonance for $t \geq \tau$.

%%%%%%
\InsertFigTwo{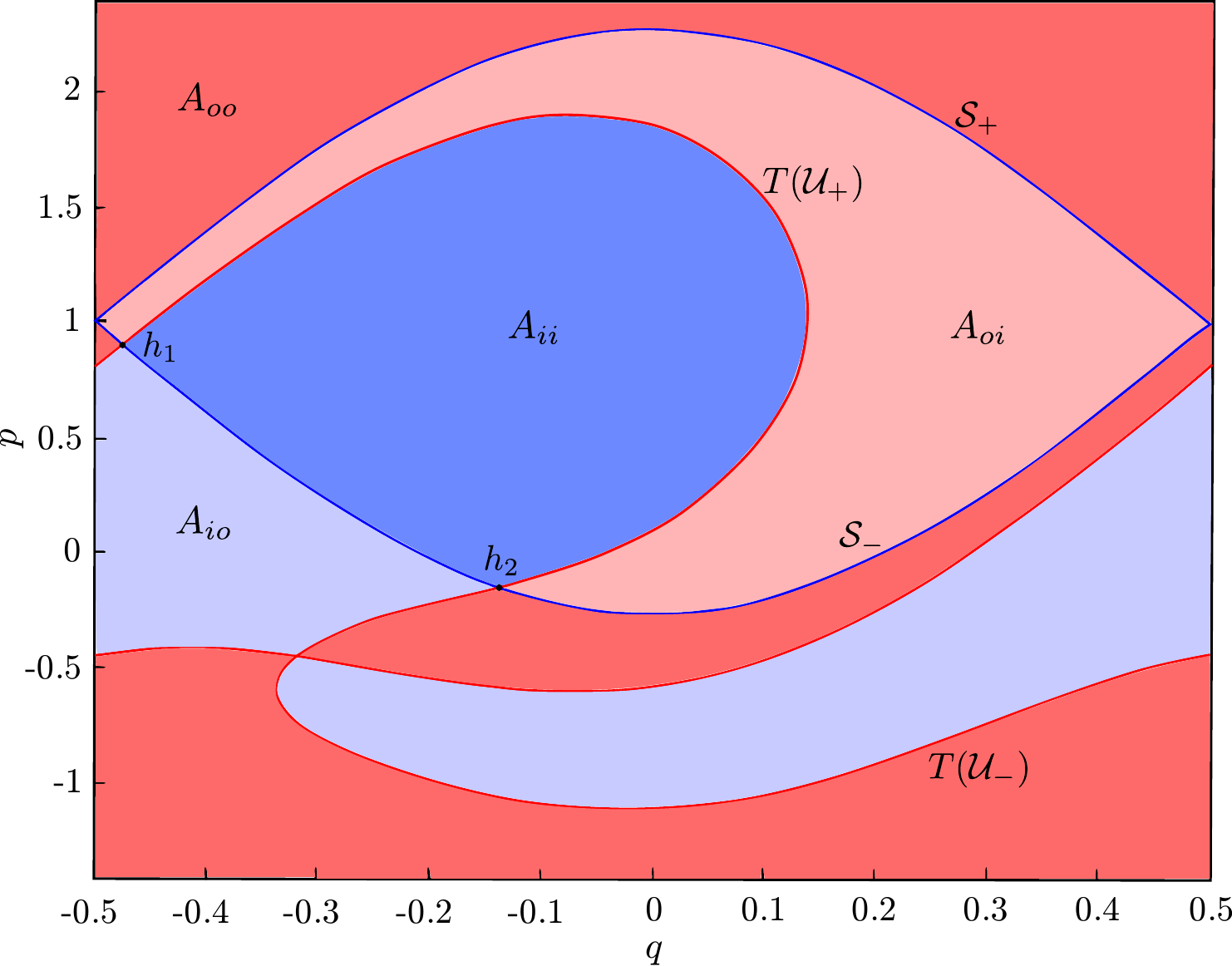}{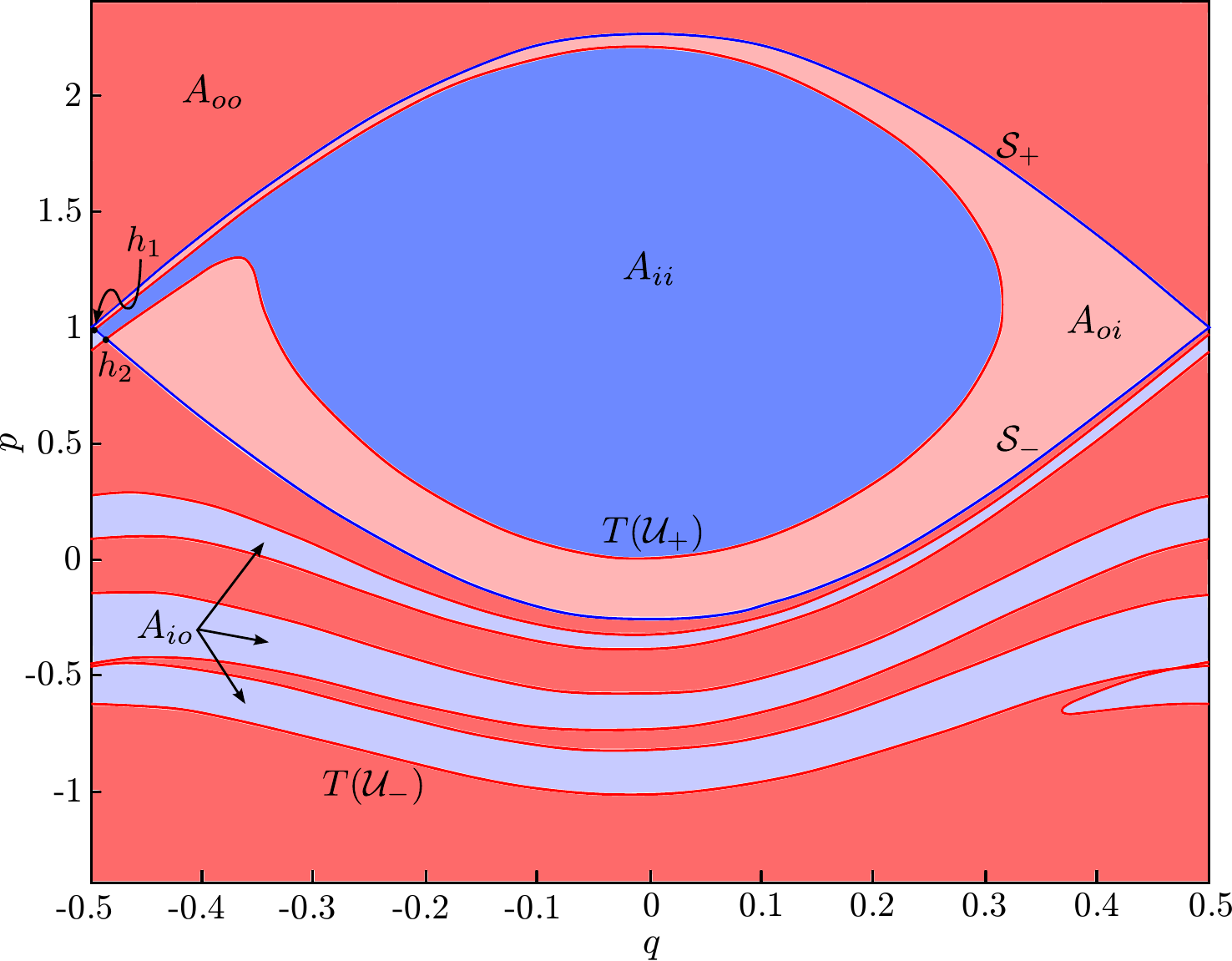}{Image of the unstable manifolds of the past resonance for \Eq{accHamiltonian} at time $\tau$ for $k = 0.4$ with transition times $\tau = 1.0$ (left) and $\tau = 3.0$ (right). The blue regions correspond to trajectories that begin inside the past resonance and the red regions to those that begin outside. The dark blue region, labeled $A_{ii}$ corresponds to those particles that remain trapped in the accelerated potential well and the light blue region, $A_{io}$ to those that are left behind.}{acc}{0.475\linewidth}
%%%%%%

In \Fig{acc} there are two heteroclinic points, $\{h_1,h_2\} \in T(\cU_+) \cap \cS_-$; however, such heteroclinic orbits do not always exist. In particular, for $\tau = 0$ the transition map is the identity and so there are heteroclinic points only when the past and future resonances overlap, that is, when $k \geq k_{crit}(0) = \frac{1}{16}$. More generally, the critical $k$ for a heteroclinic bifurcation can be computed as for the double gyre model \Eq{DoubleGyre}. The resulting curve of bifurcations, $k_{crit}(\tau)$, is shown in \Fig{kCrit} and a phase portrait at the bifurcation point for $t = \tau = 1$ is shown in the right pane of this figure. When $k > k_{crit}(\tau)$, there is a pair of heteroclinic orbits, and---unlike the double gyre---there appear to be no additional heteroclinic bifurcations as $\tau$ grows. Given the orbits of $h_1$ and $h_2$, the area $A_{ii}$ can be computed according to the simplified lobe area formula \Eq{transitoryArea}, since the system is indeed transitory. The resulting ratio $R_{acc}$ as a function of both $k$ and $\tau$ is shown in \Fig{accPercentFlux}.  

It is interesting that we never observe intersections between $T(\cU_\pm)$ and $\cS_+$. Indeed, it is easy to see that there is no flux out the top of the instantaneous separatrix. To see this, let 
\[
  E_{sep}(s) = k-\tfrac12 s^2,
\]
be the energy of the frozen separatrix and define
\beq{signFunction}
  F(q,p,s) = H(q,p,s) - E_{sep}(s) .
\eeq
Note that $F < 0$ inside the separatrix and $F > 0$ outside of it. 
%Now consider a point $(q_0, p_0, t_0) \in M \times \bR$ on the upper half of the instantaneous separatrix at some time $t_0$ and let $(q(t),p(t),t)$ be the time-$t$ slice of the orbit of this initial point.  Then,
%\[
%  \left. \frac{D}{Dt} F(q(t),p(t),t) \right|_{(q_0,p_0,t_0)} = 
%  \left.\left[ \frac{\partial F}{\partial t} + {\bf V} \cdot \nabla F  \right] \right|_{(q_0,p_0,t_0)},
%\]
%where ${\bf V} = {\bf V}(q,p,t)$ is the Hamiltonian vector field, and so employing \Eq{signFunction} gives
Differentiation along the Hamiltonian vector field gives
\[
  \dot{F} = -\dot{s}(p-s).
\]
Since $p \geq s$ for any point on the separatrix $p_+(q,s)$ and since we assume $\dot s(t) \geq 0$, then $\dot F \le 0$ on the upper separatrix. Thus the vector field never permits a trajectory to cross this separatrix from below, and a transverse intersection of $T(\cU_\pm)$ with $\cS_+$ is forbidden.

%%%%%%
\InsertFig{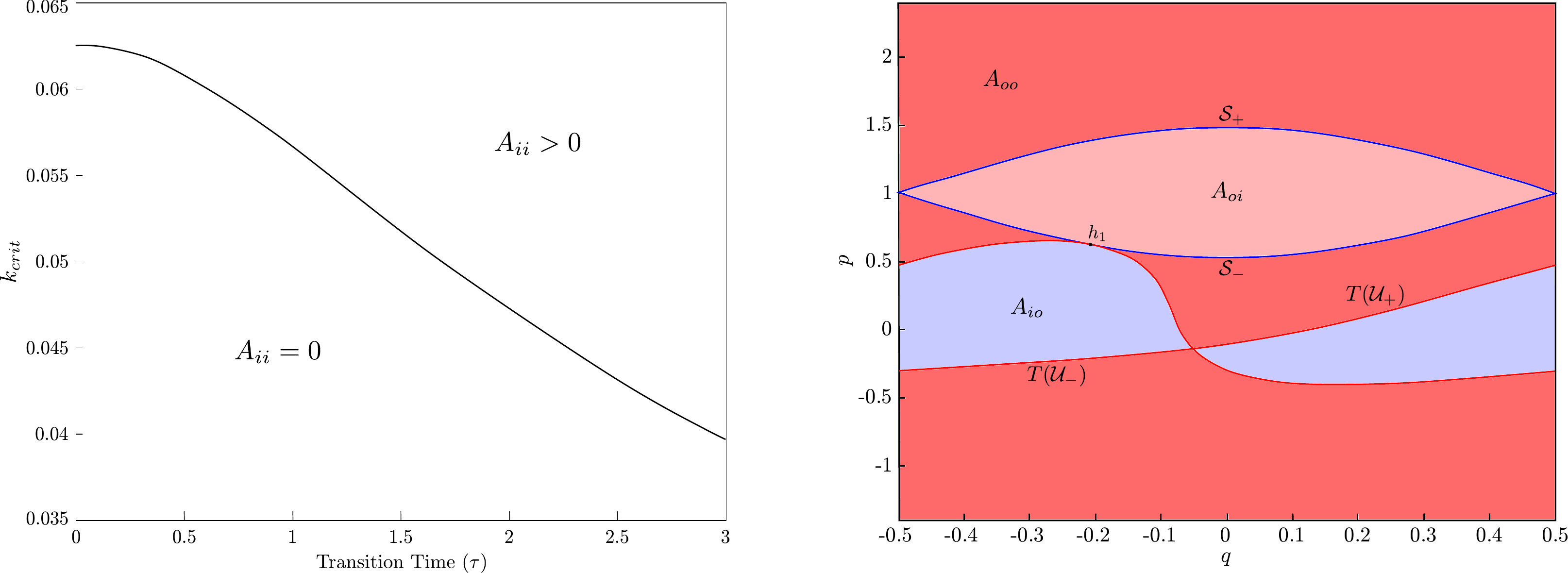}{Curve of heteroclinic bifurcations for \Eq{accHamiltonian} (left pane) and the bifurcation point for $t = \tau = 1.0$ with $k = k_{crit} \approx 0.057$ (right pane).}{kCrit}{\linewidth}
%%%%%%

%%%%%%
\InsertFig{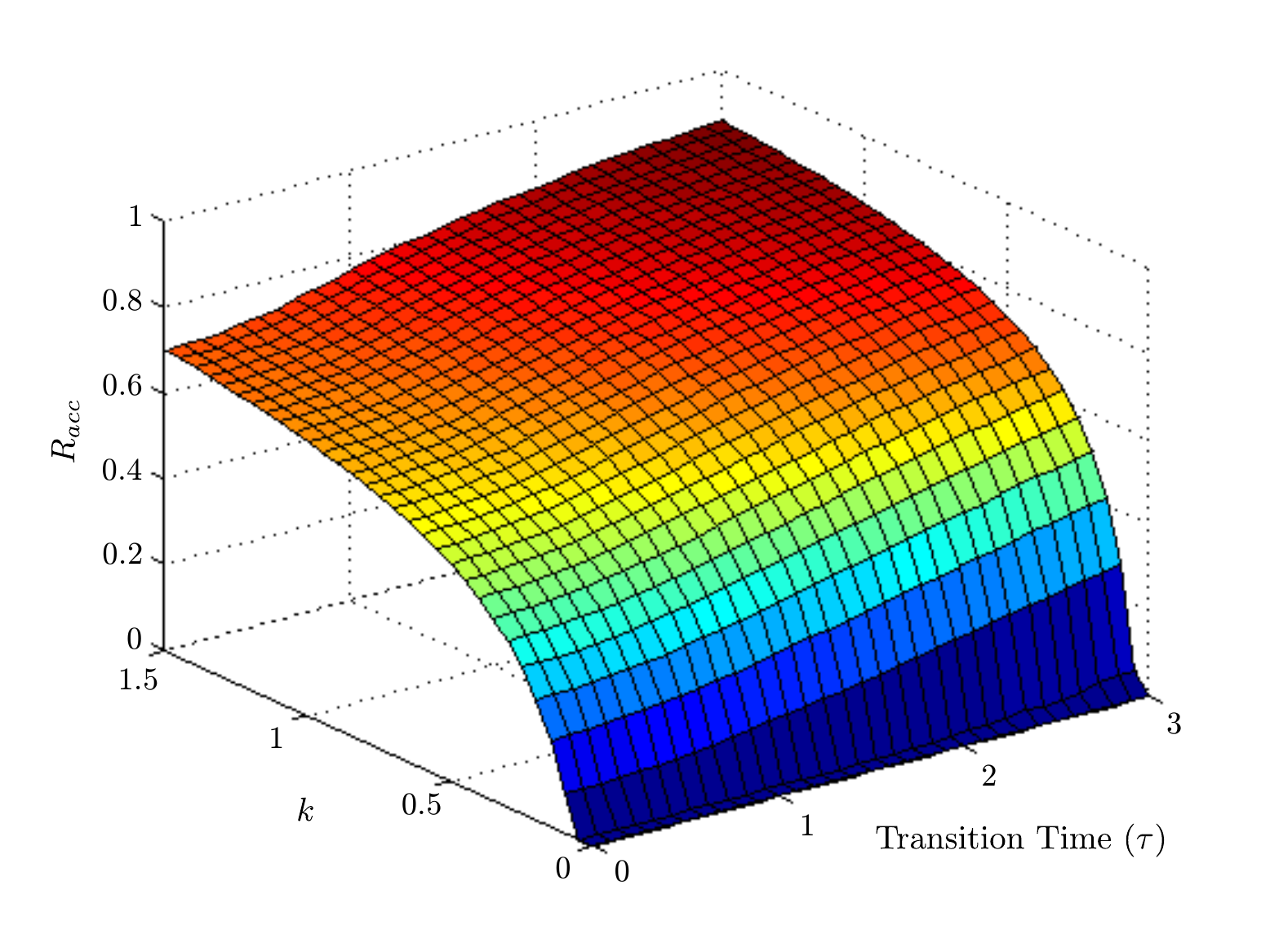}{Fraction of the past resonance that is trapped in the future resonance for \Eq{accHamiltonian} as a function of parameters $k$ and $\tau$.}{accPercentFlux}{3.5in}
%%%%%%

The unstable manifold in the right pane of \Fig{acc} can be compared with the ridges of the backward-time FTLE field in \Fig{ftle_acc}. As in the double-gyre example, the qualitative agreement is quite good, even though secondary ridges do exist. A simple threshold filtering of the FTLE field, shown in the right pane of the figure, extracts the most prominent ridge; however, this set is not a curve on the scale of the computational grid, especially in the interior and near the boundary of the future resonance. Once again, the secondary ridges are so similar in height to the main ridge that they can not be removed by a height filter.  This ambiguity as to the true location of the unstable manifold makes it hard to identify the heteroclinic points and accurately compute the desired flux from the manifolds obtained from the FTLE field.

%%%%%%
\InsertFig{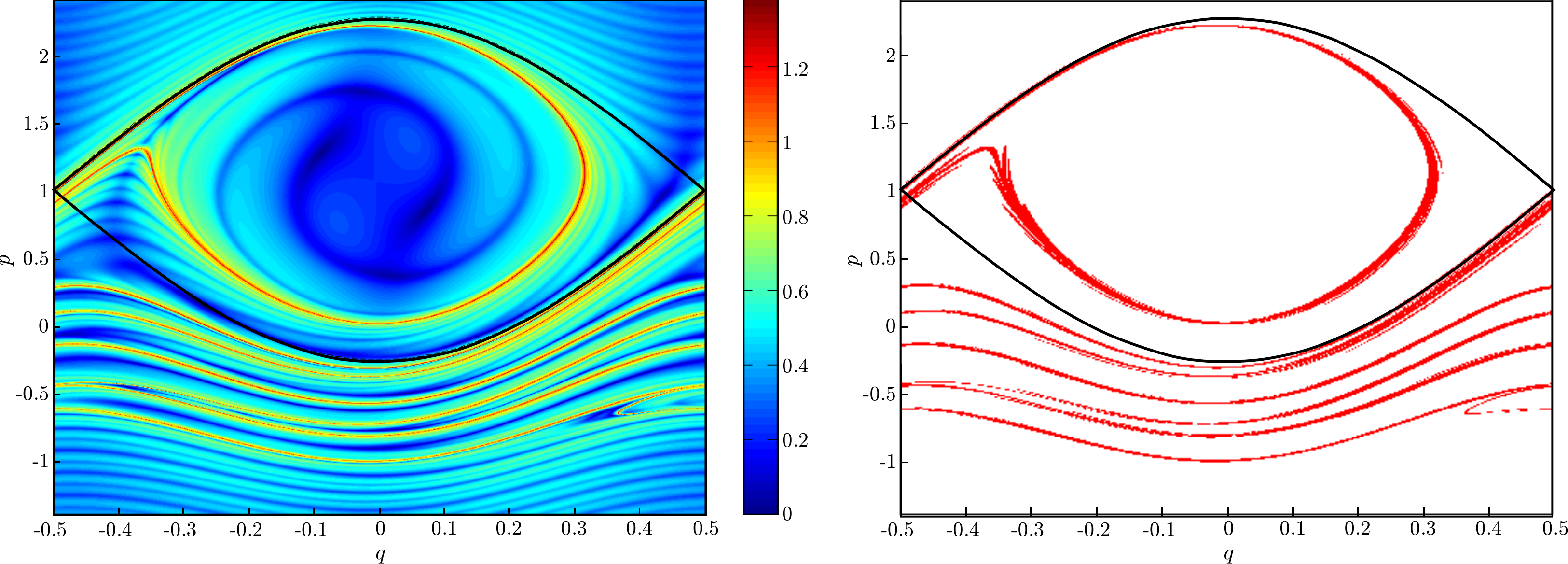}{ (left) Backward-time FTLE field for \Eq{accHamiltonian} with $k=0.4$, $\tau = 3.0$ and integration time $t=-6$. The black line denotes the resonance of the future vector field $F$.  (right) The main ridge extracted from FTLE field by keeping only those values within 45\% of the maximum.}{ftle_acc}{\linewidth}
%%%%%%

Adiabatic theory applies to the resonant accelerator system once it is written in the form \Eq{Adiabatic}. There are two topologically distinct loops in the frozen-time phase space corresponding to trapped (oscillatory) and to untrapped (rotational) trajectories. For the trapped case the loop action is the area enclosed, but for the untrapped orbits the loop action is the area contained between the graph of the trajectory and the $q$-axis.
Since $k$ is constant in our model,  the area of the trapped region of the frozen system is independent of $s$, and thus for each trapped orbit of $P$, there is a family, $\gamma_s$, of periodic orbits of the frozen systems with the same action.  Each trapped orbit that is a bounded distance inside the separatrices \Eq{accSeparatrix} has a period that is bounded; consequently, these orbits will be adiabatic in the limit $\tau \to \infty$. For the untrapped orbits, this is no longer always true. As $s$ varies, the upper separatrix of the frozen system ``sweeps through" the region bounded by the separatrix $\cU_+$, with action $J(\cU_+) = \frac{2}{\pi}\sqrt{2k}$, and the separatrix $\cS_+$, with action $J(\cS_+) = J(\cU_+)+1$. Every rotational family of orbits of the frozen system within this \emph{separatrix-swept region}, namely those with actions
\[
	J \in [J(\cU_+), J(\cS_+)],
\]
necessarily crosses the separatrix $p_+(q,s)$ for some $s \in [0,1]$. Since this implies the period is unbounded, adiabatic theory does not apply to these orbits. 

By contrast the rotational orbits with $J < J(\cU_-)$ or $J > J(\cS_+)$ remain bounded away from the frozen separatrices for all $s$, and thus are adiabatic in the limit $\tau \to \infty$. This is supported by the computations in \Fig{accAdiabatic} and the corresponding movie linked in its caption. The left pane of the figure shows eight orbits of $P$, and the middle and right panes show the images of these under $T$ for two values of $\tau$. When $\tau = 10$, each of the orbits---except for the green orbit in the separatrix swept region---has an image under $T$ that is very close to an orbit of $F$ with the same loop action.

%%%%%%
\InsertFig{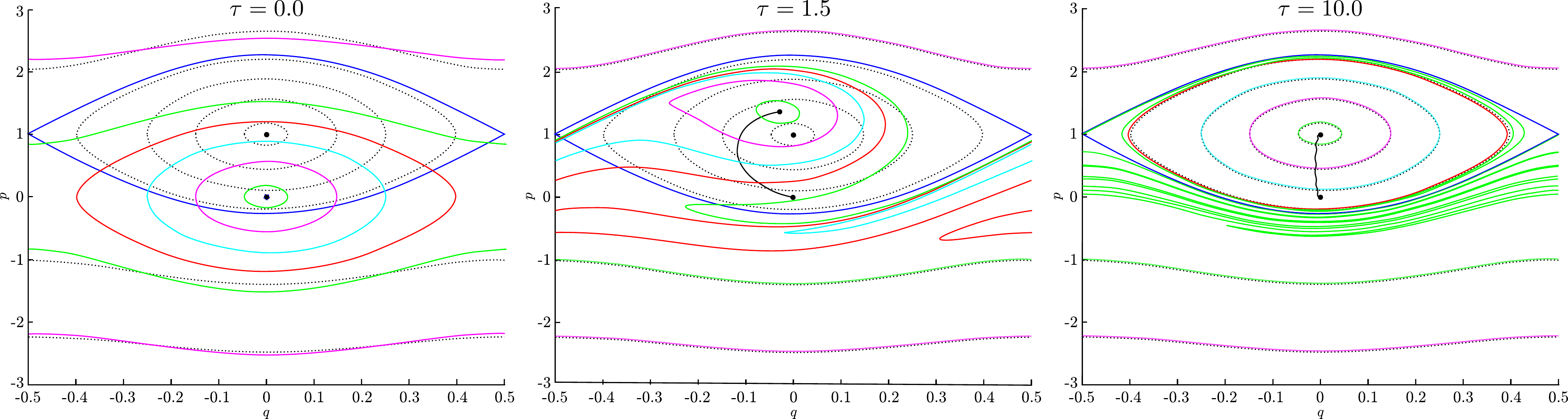}{Illustration of adiabatic invariance of the loop actions of periodic orbits for \Eq{accHamiltonian}.  The left pane shows eight orbits $\gamma^i_0$ of the past vector field, and the middle and right panes depict $T(\gamma^i_0)$ for $\tau = 1.5$ and $10$.  The blue curve is the boundary $\cS_\pm$ of the future resonance and the dotted curves indicate orbits of $F$ with the same actions as the orbits $\gamma^i_0$. The black curve shows the orbit of the elliptic equilibrium of $P$. A movie showing the variation of these images with $\tau$ is at [LINK MOVIE ``AccAdiabatic.mov" HERE]}{accAdiabatic}{\linewidth}
%%%%%%

The frozen-time accelerator model has a family of hyperbolic saddles, and each saddle has a pair of homoclinic loops. Thus, the theory of \cite{Kaper91} described in \Sec{Adiabatic} applies. Since $J_{max} = J_{min}$ for this model, the area of the lobe, $A_{io}$, limits to zero, and so this theory predicts that
\[
  \lim_{\tau \to \infty} R_{acc} = 1.
\]
That is to say, the region of area $A_{ii}$ limits to the future resonance, as \Fig{acc} and the accompanying movie seem to suggest.  In terms of the physical model this implies that, provided the acceleration is ``slow enough,'' almost all particles beginning in the resonance at $t=0$ will be stably accelerated.

%As with \Eq{DoubleGyre}, the system \Eq{accHamiltonian} has a time-reversal symmetry if the transition function $s(t)$ satisfies \Eq{transitionSymmetry}.  In this case the reversor is given by 
%\[
%  R: M \times \bR \to M \times \bR : (x,y,t) \mapsto (x, \omega-y, \tau-t),
%\]
%such that the vector field $V$ is again reversed by the push-forward $R_*$ of the reversor:
%\[
%  DRV(x, \omega-y, \tau-t) = -V(x,y,t).
%\]
%Then $T^{-1} = R \circ T \circ R$, and we can obtain the time-0 slice of $\gamma(S,\tau)$ by simply reflecting the unstable manifold at $t=\tau$, namely
%\[
%  T^{-1}(\cS) = \gamma_0(\cS,\tau) = R(T(\cU)).
%\]

%%%%%%%%%%%%%%%%%
%% Conclusions
%%%%%%%%%%%%%%%%%
\section{Conclusions}
While techniques involving finite-time Lyapunov exponents and distinguished hyperbolic trajectories have recently been developed for the identification and extraction of coherent structures in time-dependent systems, they have been used only selectively to give quantitative descriptions of the finite-time flux between such structures.  One reason for this is that the  ``ridges'' of the FTLE field that represent approximate invariant manifolds are often difficult to extract, making precise measurements of flux challenging.  

Here we have considered a special class of two-dimensional nonautonomous systems that exhibit time-dependent behavior only on a compact interval, and have extensively used the concepts of backward and forward hyperbolicity for these transitory systems. The special structure of these systems, leads to a simple a method for the numerical computation of flux between Lagrangian coherent structures in the Hamiltonian case.  Our method relies primarily on knowledge of heteroclinic orbits and their associated invariant manifolds that bound lobes within the extended phase space.  Thus, our computations of flux require very little Lagrangian information relative to computations involving FTLE or distinguished hyperbolic trajectories.  In particular, our adaptive computation of $T(U)$ allowed for an order of magnitude reduction in the number of particle advections required for a computation of the FTLE field at similar resolution.

An important extension to the theory presented here arises in light of recent advances in the theory of finite-time manifolds \cite{Haller98, Sandstede00, Duc08, Yagasaki08}.  These studies have shown that, given a system whose behavior is unknown outside an interval $\cI = [t_-, t_+]$, the manifold structure of a ``sufficiently slowly'' moving orbit at some time $t \in \textrm{int}(\cI)$ is ``unique'' up to an exponentially small correction term, provided $t$ is ``sufficiently far'' from the endpoints of $\cI$.  Since the formulas \Eq{DeltaAMinus} and \Eq{DeltaAPlus} depend only on heteroclinic orbits lying on the boundary of a lobe, they could be used directly to provide exponentially accurate approximations of lobe areas in such systems.

It is not obvious if our technique can be applied more generally to nonautonomous systems or to systems defined by a discrete set of data; however, there are several logical extensions that we plan to address in future work. The first is the case of transitory, symplectic maps, where the action formulas that we have developed should also apply. Similarly, since action formulas for flux have been recently developed in the $n$-dimensional volume preserving case \cite{Lomeli09b}, transitory volume-preserving systems could also be treated. Finally, it is reasonable that if the time dependence is ``episodic" in nature, each transition could at least approximately be treated by the same methods that we have used.\\

\bigskip
\appendix
\begin{center}
	{\bf Appendices}
\end{center}

%%%%%%%%%%%%%%%%%
%% Transition Functions
%%%%%%%%%%%%%%%%%
\section{Transition functions}\label{app:transitionFunctions}

We can obtain a $C^k$ transition function that is polynomial on $t \in [0,1]$, by requiring
$s(0)=0$, $s(1) = 1$ and $D^j s(0) = D^j s(1) = 0$ for $j=1\ldots k$.
On the interval $[0,1]$ these functions are
\bsplit{polynomial}
	s_0(t) &= t \\
	s_1(t) &= t^2(3-2t) \\
	s_2(t) &= t^3( 10-15t+6t^2) \\ 
	s_3(t) &= t^4(35-84t+70t^2-20t^3) \\
	s_4(t) &= t^5(126-420t+540t^2-315t^3+70t^4).
\end{split}\eeq
Figure \ref{fig:transitionFunctions} shows several of these functions for various values of $k$.  It is not hard to see that in general these polynomials are given by
\[
	s_{k}(t) = \frac{\Gamma(2k)}{(\Gamma(k))^2} \int_0^t s^k(1-s)^k \,ds ,
\]
for $t \in [0,1]$, which is monotone.

%%%%%
\InsertFig{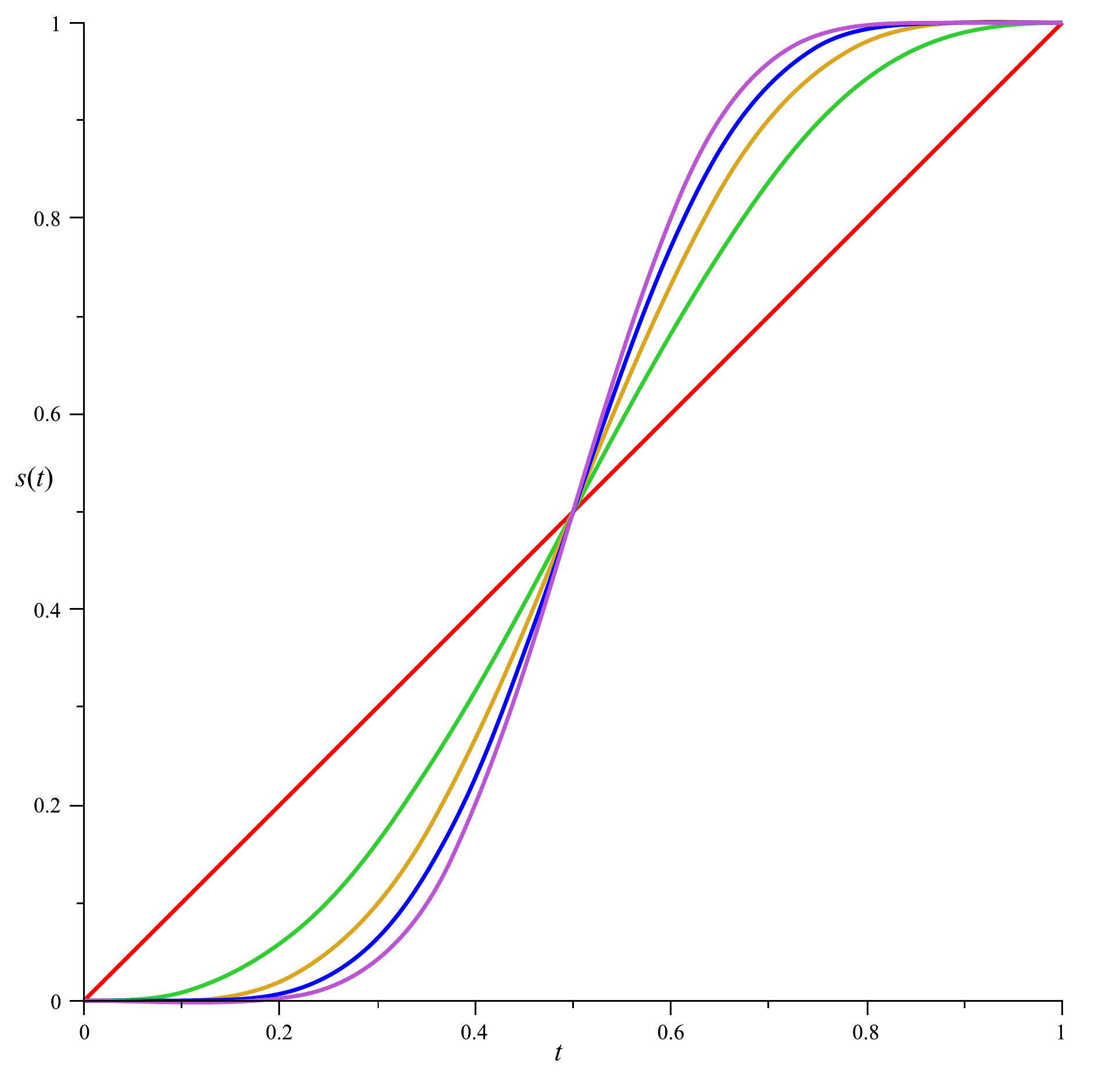}{The transition functions $s_k(t)$ for odd $k$ up to $9$.}{transitionFunctions}{2.5in}

%%%%%%%%%%%%%%%%%
%% Lie Derivatives
%%%%%%%%%%%%%%%%%
\section{Forms and Lie Derivatives}\label{app:notation}
Here we set out our notation, which follows \cite{Abraham94}. We denote the set of $k$-forms on a manifold $M$ by $\Lambda^k(M)$, and the set of vector fields by $\cV(M)$. If $\alpha \in \Lambda^k(M)$ and $V_1, V_2, \ldots V_k \in \cV(M)$, then the pullback, $f^*$, of a form $\alpha$ by a diffeomorphism $f$ is defined by
\beq{pullbackForm}
	(f^*\alpha)_x(V_1,V_2,...,V_k) =\alpha_{f(x)}(Df(x)V_1(x),\ldots, Df(x)V_k(x)) .
\eeq
The pullback can be applied to a vector field $V$ as well:
\beq{pullbackV}
	(f^{*}V)(x) = (Df(x))^{-1}V(f(x)) .
\eeq
The push-forward operator is defined as
\beq{pushforwardV}
	f_* = (f^{-1})^* .
\eeq
The interior product of $\alpha$ with $V$ is defined
as the $(k-1)$-form
\beq{innerProduct}
	\imath_V \alpha \equiv \alpha(V,\cdot,\ldots,\cdot) .
\eeq
Suppose that $\vphi_{t,t_0}$ is the ($C^1$) flow of a vector field $V(x,t)$, so that $\vphi_{t_0,t_0}(x) = x$,
and $\frac{d}{dt} \vphi_{t,t_0}(x) = V(\vphi_{t,t_0}(x),t)$.
Then the Lie derivative
with respect to $V$ is the linear operator defined by
\beq{LieDerivative}
	\cL_V \cdot \equiv  \left. \frac{d}{ds}\right|_{s=t}  \vphi_{s,t}^* \cdot
\eeq
where $\cdot$ is any tensor. In particular for a vector field $X$,
\beq{LieBracket}
	\cL_V X = [V,X] = (V \cdot \nabla) X - (X \cdot \nabla) V ,
\eeq
where $[ \;,\; ]$ is the Lie bracket. The Lie derivative acting on
differential forms obeys Cartan's homotopy formula
\beq{Cartan}
	\cL_V \alpha \equiv \imath_V ( d \alpha) + d(\imath_V \alpha) .
\eeq
Note that $\cL$ behaves ``naturally" with respect to the pullback:
\beq{naturally}
	f^*\cL_V \alpha = \cL_{f^*V} f^*\alpha .
\eeq

%%%%%
\bigskip
\bibliographystyle{siam}
\bibliography{transitory}

\begin{thebibliography}{10}

\bibitem{Abraham94}
{\sc R.~Abraham and J.E. Marsden}, {\em Foundations of Mechanics}, Benjamin,
  Reading, 2nd~ed., 1994.

\bibitem{Anderson06}
{\sc P.D. Anderson, D.J. Ternet, G.~W.~M. Peters, and H.~E.~H. Meijer}, {\em
  Experimental/numerical analysis of chaotic advection in a three-dimensional
  cavity flow}, Int. Polymer Processing, 04 (2006), pp.~412--420.

\bibitem{Aref02}
{\sc H.~Aref}, {\em The development of chaotic advection}, Phys. Fluids, 14
  (2002), pp.~1315--1325.

\bibitem{Balasuriya06}
{\sc S.~Balasuriya}, {\em Cross-separatrix flux in time-aperiodic and
  time-impulsive flows}, Nonlinearity, 19 (2006), pp.~2775--2795.

\bibitem{Beigie91a}
{\sc D.~Beigie, A.~Leonard, and S.~Wiggins}, {\em Chaotic transport in the
  homoclinic and heteroclinic tangle regions of quasiperiodically forced
  two-dimensional dynamical systems}, Nonlinearity, 4 (1991), pp.~775--819.

\bibitem{Branicki10}
{\sc M.~Branicki and S.~Wiggins}, {\em Finite-time {L}agrangian transport
  analysis: stable and unstable manifolds of hyperbolic trajectories and
  finite-time {L}yapunov exponents}, Nonlin. Proc. Geophys., 17 (2010),
  pp.~1--36.

\bibitem{Brunton10}
{\sc S.~Brunton and C.~Rowley}, {\em Fast computation of finite-time {L}yapunov
  exponent fields for unsteady flows}, Chaos, 20 (2010), p.~017503.

\bibitem{Cardwell08}
{\sc B.~Cardwell and K.~Mohseni}, {\em Vortex shedding over a two-dimensional
  airfoil: Where the particles come from}, AIAA Journal, 46 (2008),
  pp.~545--547.

\bibitem{Cary86}
{\sc J.R. Cary, D.F. Escande, and J.L. Tennyson}, {\em Adiabatic-invariant
  change due to separatrix crossing}, Phys. Rev. A, 34 (1986), pp.~4256--4275.

\bibitem{Chien86}
{\sc W.~L. Chien, H.~Rising, and J.~M. Ottino}, {\em Laminar mixing and chaotic
  mixing in several cavity flows}, J. Fluid Mech., 170 (1986), pp.~355--377.

\bibitem{Coulliette01}
{\sc C.~Coulliette and S.~Wiggins}, {\em Intergyre transport in a wind-driven,
  quasigeostrophic double gyre: An application of lobe dynamics}, Nonlin. Proc.
  Geophys., 8 (2001), pp.~69--94.

\bibitem{Duc08}
{\sc L.~H. Duc and S.~Siegmund}, {\em Hyperbolicity and invariant manifolds for
  planar nonautonomous systems on finite time intervals}, Internat. J. Bifur.
  Chaos Appl. Sci. Engrg., 18 (2008), pp.~641--674.

\bibitem{Easton91}
{\sc R.W. Easton}, {\em Transport through chaos}, Nonlinearity, 4 (1991),
  pp.~583--590.

\bibitem{Edwards04}
{\sc D.A. Edwards and M.J. Syphers}, {\em An Introduction to the Physics of
  High Energy Accelerators}, Wiley, New York, 2004.

\bibitem{Elskens91}
{\sc Y.~Elskens and D.~F. Escande}, {\em Slowly pulsating separatrices sweep
  homoclinic tangles where islands must be small: an extension of classical
  adiabatic theory}, Nonlinearity, 4 (1991), pp.~615--667.

\bibitem{Froyland09}
{\sc G.~Froyland and K.~Padberg}, {\em Almost-invariant sets and invariant
  manifolds---connecting probabilistic and geometric descriptions of coherent
  structures in flows}, Phys. D, 238 (2009), pp.~1507--1523.

\bibitem{GGTH07}
{\sc C.~Garth, F.~Gerhardt, X.~Tricoche, and H~Hagen}, {\em Efficient
  computation and visualization of coherent structures in fluid flow
  applications}, IEEE Trans. Vis. and Comput. Graph., 13 (2007),
  pp.~1464--1471.

\bibitem{Haller01}
{\sc G.~Haller}, {\em Distinguished material surfaces and coherent structures
  in three-dimensional flows}, Phys. D, 149 (2001), pp.~248--277.

\bibitem{Haller97}
{\sc G.~Haller and A.~C. Poje}, {\em Eddy growth and mixing in mesoscale
  oceanographic flows}, Nonlin. Proc. Geophys., 4 (1997), pp.~223--235.

\bibitem{Haller98}
\leavevmode\vrule height 2pt depth -1.6pt width 23pt, {\em Finite time
  transport in aperiodic flows}, Phys. D, 119 (1998), pp.~352--380.

\bibitem{Haller00}
{\sc G.~Haller and G.~Yuan}, {\em Lagrangian coherent structures and mixing in
  two-dimensional turbulence}, Phys. D, 147 (2000), pp.~352--370.

\bibitem{Hobson93}
{\sc D.~Hobson}, {\em An efficient method for computing invariant manifolds of
  planar maps.}, J. Comput. Phys., 104 (1993), pp.~14--22.

\bibitem{Ide02}
{\sc K.~Ide, D.~Small, and S.~Wiggins}, {\em Distinguished hyperbolic
  trajectories in time-dependent fluid flows: analytical and computational
  approach for velocity fields defined as data sets}, Nonlin. Proc. Geophys., 9
  (2002), pp.~237--263.

\bibitem{Madrid09}
{\sc J.~Jim\'enez~Madrid and A.~Mancho}, {\em Distinguished trajectories in
  time dependent vector fields}, Chaos, 19 (2009), p.~013111.

\bibitem{Kaper91}
{\sc T.J. Kaper and S.~Wiggins}, {\em Lobe area in adiabatic {H}amiltonian
  systems}, Phys. D, 51 (1991), pp.~205--212.

\bibitem{Kaper92}
\leavevmode\vrule height 2pt depth -1.6pt width 23pt, {\em On the structure of
  separatrix-swept regions in singularly-perturbed {H}amiltonian systems},
  Differential Integral Equations, 5 (1992), pp.~1363--1381.

\bibitem{Khakhar87}
{\sc D.~V. Khakhar, J.~G. Franjione, and J.~M. Ottino}, {\em A case study of
  chaotic mixing in deterministic flows: the partitioned-pipe mixer}, Chem.
  Engrg. Sci., 42 (1987), pp.~2909--2926.

\bibitem{Knobloch87}
{\sc E.~Knobloch and J.B. Weiss}, {\em Chaotic advection by modulated traveling
  waves}, Phys. Rev. A, 36 (1987), pp.~1522--1524.

\bibitem{Kruskal62}
{\sc M.~Kruskal}, {\em Asymptotic theory of {H}amiltonian and other systems
  with all solutions nearly periodic}, J. Math Phys., 3 (1962), pp.~806--828.

\bibitem{Lekien07}
{\sc F.~Lekien, S.~Shadden, and J.~Marsden}, {\em {L}agrangian coherent
  structures in \textit{n}-dimensional systems}, J. Math. Phys., 48 (2007),
  p.~065404.

\bibitem{Lenz09}
{\sc F.~Lenz, C.~Petri, F.R.N. Koch, F.K. Diakonos, and P.~Schmelcher}, {\em
  Evolutionary phase space in driven elliptical billiards}, New J. Phys., 11
  (2009), p.~080305.

\bibitem{Leong89}
{\sc C.W. Leong and J.M. Ottino}, {\em Experiments on mixing due to chaotic
  advection in a cavity}, J. Fluid Mech., 209 (1989), pp.~463--499.

\bibitem{Lipinski10}
{\sc D.~Lipinski and K.~Mohseni}, {\em A ridge tracking algorithm and error
  estimate for efficient computation of {L}agrangian coherent structures},
  Chaos, 20 (2010), p.~017504.

\bibitem{Lomeli09b}
{\sc H.~Lomel\'i and J.D. Meiss}, {\em Resonance zones and lobe volumes for
  exact volume-preserving maps}, Nonlinearity, 22 (2009), pp.~1761--1789.

\bibitem{MacKay91}
{\sc R.S. MacKay}, {\em A variational principle for odd dimensional invariant
  submanifolds of an energy surface for {H}amiltonian systems}, Nonlinearity, 4
  (1991), pp.~155--157.

\bibitem{MM86}
{\sc R.S. MacKay and J.D. Meiss}, {\em Flux and differences in action for
  continuous time {H}amiltonian systems}, J. Phys. A, 19 (1986), pp.~225--229.

\bibitem{MM88}
\leavevmode\vrule height 2pt depth -1.6pt width 23pt, {\em Relationship between
  quantum and classical thresholds for multiphoton ionization of excited
  atoms}, Phys. Rev. A, 37 (1988), pp.~4702--4706.

\bibitem{MMP84}
{\sc R.S. MacKay, J.D. Meiss, and I.C. Percival}, {\em Transport in
  {H}amiltonian systems}, Phys. D, 13 (1984), pp.~55--81.

\bibitem{MMP87}
{\sc R.~S. MacKay, J.~D. Meiss, and I.~C. Percival}, {\em Resonances in
  area-preserving maps}, Phys. D, 27 (1987), pp.~1--20.

\bibitem{Malhotra98}
{\sc N.~Malhotra and S.~Wiggins}, {\em Geometric structures, lobe dynamics, and
  {L}agrangian transport in flows with aperiodic time-dependence, with
  applications to {R}ossby wave flow}, J. Nonlinear Sci., 8 (1998),
  pp.~401--456.

\bibitem{Markus53}
{\sc L.~Markus}, {\em Asymptotically autonomous differential systems}, in
  Contributions to the Theory of Nonlinear Oscillations, Solomon Lefschetz,
  ed., vol.~3 of Annals of Mathematics Studies, Princeton Univ. Press,
  Princeton, 1953, pp.~17--29.

\bibitem{Mathur07}
{\sc M.~Mathur, G.~Haller, T.~Peacock, J.E. Ruppert-Felsot, and H.L. Swinney},
  {\em Uncovering the {L}agrangian skeleton of turbulence}, Phys. Rev. Lett.,
  98 (2007), p.~144502.

\bibitem{Mendoza10}
{\sc C.~Mendoza, A.~Mancho, and M.~Rio}, {\em The turnstile mechanism across
  the {K}uroshio current: analysis of dynamics in altimeter velocity fields},
  Nonlin. Proc. Geophys., 17 (2010), pp.~103--111.

\bibitem{Miller97}
{\sc P.D. Miller, C.K.R.T. Jones, and L.J. Pratt}, {\em Quantifying transport
  in numerically generated velocity fields}, Physica D, 110 (1997), pp.~1--18.

\bibitem{Neishtadt87}
{\sc A.~I. Ne{\u\i}shtadt}, {\em On the change in the adiabatic invariant on
  crossing a separatrix in systems with two degrees of freedom}, J. Appl. Math.
  Mech., 51 (1987), pp.~586--592.

\bibitem{Poje99}
{\sc A.~C. Poje and G.~Haller}, {\em Geometry of cross-stream mixing in a
  double-gyre ocean model}, J. Phys. Oceanogr., 29 (1999), pp.~1649--1665.

\bibitem{Rogerson99}
{\sc A.~M. Rogerson, P.~D. Miller, L.~J. Pratt, and C.K.R.T. Jones}, {\em
  Lagrangian motion and fluid exchange in a barotropic meandering jet}, J.
  Phys. Oceanogr., 29 (1999), pp.~2635--2655.

\bibitem{RomKedar90c}
{\sc V.~Rom-Kedar and S.~Wiggins}, {\em Transport in two-dimensional maps},
  Arc. Rational Mech. Anal., 109 (1990), pp.~239--298.

\bibitem{SP07}
{\sc F~Sadlo and R~Peikert}, {\em Efficient visualization of {L}agrangian
  coherent structures by filtered {AMR} ridge extraction}, IEEE Trans. Vis. and
  Comput. Graph., 13 (2007), pp.~1456--1463.

\bibitem{Samelson07}
{\sc R.M. Samelson and S.~Wiggins}, {\em Lagrangian transport in geophysical
  jets and waves : the dynamical systems approach}, vol.~31 of
  Interdisciplinary Applied Mathematics, Springer, New York, 2007.

\bibitem{Sandstede00}
{\sc B.~Sandstede, S.~Balasuriya, C.~K. R.~T. Jones, and P.~Miller}, {\em
  Melnikov theory for finite-time vector fields}, Nonlinearity, 13 (2000),
  pp.~1357--1377.

\bibitem{Shadden05}
{\sc S.C. Shadden, F.~Lekien, and J.~Marsden}, {\em Definition and properties
  of {L}agrangian coherent structures from finite-time {L}yapunov exponents in
  two-dimensional aperiodic flows}, Phys. D, 212 (2005), pp.~271--304.

\bibitem{Weiss89}
{\sc J.B. Weiss and E.~Knobloch}, {\em Mass transport and mixing by modulated
  traveling waves}, Phys. Rev. A, 40 (1989), pp.~2579--2581.

\bibitem{Yagasaki08}
{\sc K.~Yakasaki}, {\em Invariant manifolds and control of hyperbolic
  trajectories on infinite- or finite-time intervals}, Dyn. Syst., 23 (2008),
  pp.~309--331.

\end{thebibliography}

\end{document}